\documentclass[journal,onecolumn,12pt]{IEEEtran} %
\usepackage[utf8]{inputenc}

\title{ Deployment Optimization of Dual-functional UAVs for Integrated Localization and Communication}
\author{Zheyuan Yang,~\IEEEmembership{Graduate Student Member,~IEEE}, Suzhi Bi,~\IEEEmembership{Senior Member,~IEEE}, Ying-Jun Angela Zhang,~\IEEEmembership{Fellow,~IEEE} \thanks{ Z.~Yang and Y-J.~A.~Zhang are with the Department of Information Engineering, The Chinese University of Hong Kong, Hong Kong (\{yz019, yjzhang\}@ie.cuhk.edu.hk). S. Bi is with the College of Electronics and Information Engineering, Shenzhen University, Shenzhen, China (bsz@szu.edu.cn). S. Bi is also with the Peng Cheng Laboratory, Shenzhen, China. } }
\usepackage{cite}
\usepackage{graphicx}

\usepackage{amsmath}
\usepackage{amssymb}
\usepackage{amsthm}
\usepackage{commath} 
\usepackage{mathtools}
\usepackage{subfigure}
\usepackage{bm}
\usepackage{algcompatible}
\usepackage{algorithm}
\usepackage{xcolor}
\usepackage{relsize}
\usepackage{amssymb}

\usepackage{geometry}
\theoremstyle{definition}

\newtheorem{lemma}{Lemma}
\newtheorem{proposition}{Proposition}
\newtheorem{corollary}{Corollary}

\begin{document}
\maketitle
\begin{abstract}
 In emergency scenarios, unmanned aerial vehicles (UAVs) can be deployed to assist localization and communication services for ground terminals. In this paper, we propose a new integrated air-ground networking paradigm that uses dual-functional UAVs to assist the ground networks for improving both communication and localization performance. We investigate the optimization problem of deploying the minimal number of UAVs to satisfy the communication and localization requirements of ground users. The problem has several technical difficulties including the cardinality minimization, the non-convexity of localization performance metric regarding UAV location, and the association between user and communication terminal. To tackle the difficulties, we adopt $D$-optimality as the localization performance metric, and derive the geometric characteristics of the feasible UAV hovering regions in 2D and 3D based on accurate approximation values. We solve the simplified 2D projection deployment problem  by transforming the problem into a minimum hitting set problem, and propose a low-complexity algorithm to solve it. Through numerical simulations,  we compare our proposed algorithm with benchmark methods. The number of UAVs required by the proposed algorithm is close to the optimal solution, while other benchmark methods require much more UAVs to accomplish the same task. 
 \end{abstract}


\allowdisplaybreaks
\section{Introduction}
 
  \subsection{Motivations}
Modern 5G communication networks with high speed and low latency are expected to cover most urban areas. New applications enabled by 5G such as remote healthcare, Internet of Vehicles, smart homes, industrial control, and environmental monitoring will bring unprecedented intelligent  service experience to citizens \cite{hu2021digital,liu2021promoting,kumhar2021emerging}. However, urban 5G networks are under the risk of regional service interruption in case of  earthquakes, floods or other emergency accidents\cite{hasan2021search}. Besides, in  remote areas such as mountainous environments and wild fields, it is difficult to guarantee smooth communication when performing terrain survey, reconnaissance and rescue operations due to the high deployment cost and insufficient coverage of communications infrastructure.  Recently, unmanned aerial vehicle (UAV) technology has emerged as an important solution to assist communication in emergency, where UAV-mounted aerial base stations (BSs) can be flexibly deployed to cover ground users with high-speed and reliable line-of-sight (LoS) communication services\cite{shahzadi2021uav}. 

Besides reliable communication services, ground terminals such as rescue devices also demand accurate positioning services in emergency scenarios.  The global navigation satellite system (GNSS) can provide satisfactory positioning services in open-sky environments \cite{zhang2018quality}. However, under severe obstructions and scattering in mountainous or disaster environments, the positioning accuracy of GNSS plummets dramatically (e.g., from 5-10 meters root mean square error to tens of meters) and the positioning service suffers from frequent interruptions\cite{wang2021toward}. Besides, the GNSS has relatively low positioning  accuracy in vertical dimension, and thus it cannot provide accurate three-dimensional (3D) localization services required in many emergency scenarios\cite{xiao2022overview}. In addition to GNSS,  existing terrestrial cellular networks have independent localization capabilities, where  terrestrial BSs can serve as anchor nodes (ANs) to support a variety of localization methods such as observed time difference of arrival (OTDoA) positioning \cite{fischer2014observed}. Achieving high 3D localization accuracy requires not only sufficient number of ANs, but also spatial diversity in the deployment of ANs\cite{chen2016three}. For instance, to attain high vertical positioning accuracy, the anchors should differ noticeably in amplitude. When the ground ANs are co-planar, UAVs can be quickly and flexibly deployed as temporary aerial ANs to improve the overall localization accuracy, especially in the vertical dimension. Existing UAV deployment solutions mostly use UAVs for either communication or positioning purpose, and  lack efficient coordination with terrestrial networks. As a result, a large number of UAVs are needed  to ensure both communication and positioning performance in a large area.  In practice, this could lead to large networking delay, high control complexity and signalling overhead.

To minimize the UAV deployment cost of emergency networks,  we propose a dual-functional UAV paradigm that utilizes UAVs as both air-based  ANs and BSs to provide integrated localization and communication (ILAC) services to ground UEs. Specifically, we consider to deploy multiple dual-functional UAVs working collaboratively with terrestrial BSs to extend the communication and localization service coverage. The key design problem lies in the potential conflicting effects of UAV positions on the communication and localization performance of ground UEs. Intuitively, as an aerial BS, a UAV needs to be deployed close to the ground UEs for stronger communication link quality; as an aerial AN, however, such a close-to-ground UAV deployment may cause large vertical localization error due to its similar altitude with terrestrial anchors. To meet the diverse service requirements of ground UEs, the deployment of UAVs should coordinate with terrestrial BSs to achieve a balanced communication and localization performance. 

To this end, this paper aims to answer the following two key questions:

\textbf{Q1:} For a UE at a given target location, how to determine the position of a single dual-functional UAV \emph{in an analytical form}, such that the UAV and ground BSs can collaboratively satisfy the  communication and localization performance requirements of the UE.

\textbf{Q2:}  For a set of target UEs distributed over a large area, how to deploy minimum number of dual-functional UAVs \emph{within a short computational time}, such that the integrated air-ground network can meet the communication and localization performance requirements of all the UEs.

To understand the technical challenges of the above questions, we review in the following some related work using UAVs to provide communication and localization services.

\subsection{Related Work}

\subsubsection{UAV-assisted Communication}

There have been extensive studies in recent years on employing UAVs to assist  terrestrial communications. For instance,  UAVs are deployed in mobile hot-spots to offload cellular traffic, as aerial relays to tackle obstructions in the ground, as transceivers to collect/disseminate data from/to massive IoT devices in rural areas, and as temporary aerial stations in emergency, etc.  UAV deployment is one key design problem in UAV-assisted communications. Depending on the mobility of UAVs, existing studies are mainly divided into two streams: one optimizes the fixed locations of rotary-wing UAVs and the other designs the trajectories of mobile UAVs, including both rotary-wing and fixed-wing UAVs. In this paper, we focus on the former topic to deploy dual-functional UAVs at fixed locations. 

Considering a simplified one-dimensional (1D) scenario, the authors in \cite{al2014optimal} derive the formula of LoS probability between the UAV and the ground users, and optimize the UAV altitude to maximize the UAV communication coverage. For multiple UAVs, the authors in \cite{chen2017optimum} study the optimal altitude of UAVs as relays to minimize the end-to-end outage probability and bit error rate (BER). \cite{lyu2016placement} and  \cite{ali2020uav} further consider the UAV placement for communication in 2D scenario when the UAV altitude is fixed.  \cite{lyu2016placement} aims to minimize the number of UAVs to cover a group of ground users by sequentially deploying the UAVs along a spiral path from area perimeter toward the center.\cite{ali2020uav} jointly optimizes the UAV 2D locations and user transmission power to maximize the  user communication rate to both UAVs and multiple terrestrial BSs. Furthermore, recent studies \cite{ren2020joint,fan2018optimal,sabzehali20213d} investigate how to deploy UAVs in 3D space to improve the communication performance. Considering using a single UAV as a relay, \cite{ren2020joint} minimizes the decoding error probability in latency-critical scenarios by optimizing the location and power of the UAV. Using multiple UAVs as relays, \cite{fan2018optimal} searches for the optimal UAV locations to maximize the system communication rate. \cite{sabzehali20213d} investigates the problem of determining 3D placement and orientation of the minimal number of UAVs to guarantee the LoS coverage and the signal-to-noise ratio (SNR) between the UAV-UE pairs.

  \subsubsection{UAV-assisted Localization}
  
  Besides GNSS, a UE can also estimate its location from the positioning reference signals sent by cellular networks,  such as received signal strength (RSS), time of arrival (ToA) and time difference of arrival (TDoA)\cite{8226757}. Among them, location estimation using TDoA measurements is one popular method used in 4G/5G standard for high-precision localization. The estimation is based on the principle of linear least-square estimation, where at least four ANs are required to obtain an unbiased 3D location estimation. Additional ANs can provide redundancy in the estimation to counter noisy measurements. 
  
  Given the deployment of ANs, there have been extensive studies to optimize  the resource allocation such as power and bandwidth allocation to minimize the localization error. \cite{liu2018spectrum} and \cite{wu2019resource} study the resource allocation for UAV-assisted vehicle localization using ToA method. Both studies obtain the global optimal solution based on semi-definite programming (SDP) that minimizes the Cramer-Rao lower bound (CRLB) of the localization error. However, minimizing the localization error with respect to  UAV anchor deployment is difficult due to the non-convexity of CRLB regarding the UAV position. The current research on UAV deployment and trajectory optimization mainly adopts exploration algorithms, such as genetic algorithms,   predetermined-pattern trajectory optimization and particle swarm optimization \cite{sorbelli2018range,yang2021deployment,esrafilian2020three}. \cite{sorbelli2018range} plans a static path for a single UAV by selecting a subset of way-points from a predetermined point set to maximize the localization precision. In \cite{yang2021deployment}, the authors jointly optimize the UAV positions, power and bandwidth allocation to improve the localization accuracy of vehicles by Taylor expansion-based approximate searching algorithm. Based on RSS measurements, \cite{esrafilian2020three} designs the UAV trajectory to assist 3D map estimation using particle swarm optimization. In general, the existing exploration-based methods lack theoretical analysis on the optimal UAV deployment. As a result, the obtained solution often cannot guarantee localization performance and may induce high computational complexity in emergent UAV deployment applications. 

\subsubsection{Summary}

To summarize, communication-oriented UAV placement is well-understood from 1D to 3D scenarios. However, localization-oriented UAV placement lacks theoretical characterization and efficient optimization tools. In addition, most work studies the communication and localization problems separately, which may lead to high deployment cost and delay in emergency network. Reusing one UAV for providing both communication and localization services is an effective solution to reduce the network deployment cost, which however is also more challenging considering the potential performance conflict in the dual-functional UAV deployment. To achieve guaranteed quality of service (QoS), we need both theoretical analysis and low-complexity algorithm design for the deployment of dual-functional UAVs.

\subsection{Contributions}
 
 In this paper, we study the optimal deployment problem of dual-functional UAVs that operate collaboratively with ground BSs to provide both communication and localization services to ground UEs. In particular, we are interested in deploying minimal UAVs to satisfy the communication and localization performance requirements of UEs at a set of target locations. Our contributions are summarized as follows. 
 \subsubsection{\textbf{Low-cost networking with dual-functional UAVs}} We propose a new integrated air-ground networking paradigm that uses dual-functional UAVs to assist the ground networks for improving both communication and localization services. We consider a UE using OTDoA method to estimate its own location from the positioning signals sent by a UAV and three ground BSs. Besides, each UE communicates to either a UAV or a BS that yields the highest data rate. The proposed dual-functional UAV scheme significantly reduces the cost and delay in the network deployment under emergency.  

 \subsubsection{\textbf{Minimal UAV deployment problem formulation}} We formulate a minimal UAV deployment problem that 
 optimizes both the number and locations of UAVs to meet the communication and localization requirements of ground UEs. The problem is very challenging because of the intractable cardinality minimization, the combinatorial UAV-UE communication association, and non-convex UAV location optimization. We solve the problem in two steps. First, we derive the closed-form expression of feasible UAV location to meet the dual performance requirements of each UE. Then, we transform the minimum deployment  problem into an equivalent  graphical form and design an efficient algorithm accordingly.   
  \subsubsection{\textbf{Geometric characterization of feasible UAV location}} We first analyze the feasible deployment location of a UAV that meets the localization accuracy requirement of a UE. Instead of conventional CRLB performance metric, we propose to use \textit{D-optimality} as the key localization accuracy metric for analytical tractability. We show that, under the \textit{D-optimality} metric, the feasible location of the UAV can be characterized as a second-order cone in 3D space. Meanwhile, given a fixed altitude, it reduces to an ellipse in 2D projection plane. We derive the closed-form expressions of both the 3D and 2D feasible regions of the UAV.
  \subsubsection{\textbf{Efficient graphical solution algorithm}} With the geometrical characterization of feasible region, we show that the minimal deployment problem can be equivalently transformed to a \textit{minimum hitting set} problem, which is a classical NP-complete problem that lacks efficient solution. Based on the equivalent graphical formulation, we propose a low-complexity approximate algorithm to tackle the NP-hardness of the problem in large-size emergent networks. 
  
  Simulation results show that the proposed method achieves close-to-optimal performance with much lower computational complexity. Besides, compared with conventional scheme with four ground anchors to locate a 3D target, we find  that the proposed integrated air-ground network paradigm significantly  improves the localization accuracy in both horizontal and vertical directions. This is because the deployed UAV improves not only the altitude diversity, but also the anchor-target localizing link strength especially when the NLoS noise is large in ground-to-ground channels. Overall, the proposed UAV deployment method provides an efficient solution to guarantee both communication and localization performance requirements.

 The rest of paper is organized as follows. Section \ref{sec:sysmodel} introduces the system model of UAV-assisted air-ground ILAC scheme and we formulation the minimum UAV deployment problem in Section \ref{sec:problem}. In Section \ref{sec:analysis} and Section \ref{sec:solution}, we analyze the feasible UAV hovering regions and propose a low-complexity deployment algorithm.  Section \ref{sec:simulation} presents the simulation results to evaluate the algorithm performance. Lastly, Section \ref{sec:conclusion} concludes the paper.
 
\section{System Model} \label{sec:sysmodel}

As shown in Fig. \ref{fig:system}, we consider an emergency network consisting of 3 terrestrial BSs, $U$ rotary-wing UAVs and $K$ UEs. The locations of the $n$-th BS, the $u$-th UAV and the $k$-th UE are denoted as $\bm b_n=[x_n^b, y_n^b,h_n^b,], n \in \mathcal N \triangleq \{1,2,3\}$, $\bm b_u=[x_u^a,y_u^a,h_u^a],\forall u \in \mathcal U \triangleq \{1,2,...,U\}$ and $\bm m_k=[x_k,y_k,h_k], \forall k \in \mathcal K \triangleq \{1,2,...,K\}$, respectively.\footnote{Notice that the accurate user location is not known and to be estimated. The location $\bm m_k$ is in fact a target user location to receive communication and localization service, e.g., from initial coarse location estimation.  A naive way to set the target user locations is dividing the area of interest into equally separated points. Without causing confusions, we use target user location and user location interchangeably in this paper.} The  UAVs and BSs collaboratively provide communication and localization services to ground UEs.

\begin{figure}[ht]
	\centering
	\includegraphics[width=0.5\textwidth ]{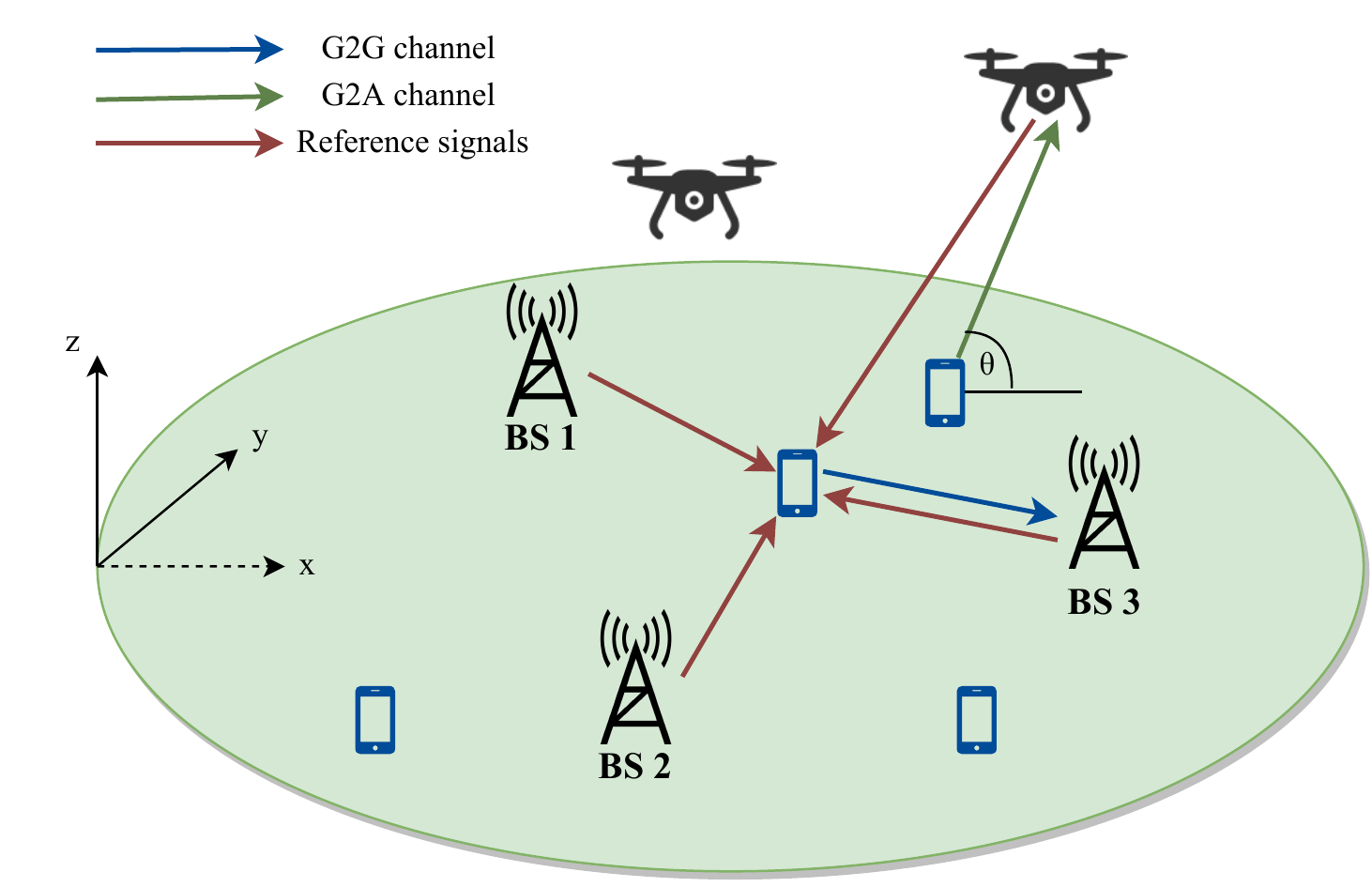}
	\caption{System model of the air-ground ILAC network.}
	\label{fig:system}
\end{figure}

\subsection{Communication Model}

 Considering the occasional blockage between the UAVs and the ground UEs, we adopt the commonly used probabilistic LoS channel model to determine the large-scale attenuation for the ground-to-air (G2A) links\cite{al2014optimal}. The probability of geometric LoS between the UAV and UE depends on statistical parameters related to the environment and the elevation angle. Specifically, we denote the LoS probability between the UAV at location $\bm b_u$ and UE $k$ as $\mathbb{P}(LoS,\theta_{ku})$, which can be approximated as a modified sigmoid function of the following form \cite{al2014optimal} 
 \begin{equation}
 	\mathbb{P}(LoS,\theta_{ku})=\tfrac{1}{1+e_1\exp \left(-e_2(\theta_{ku}-e_1 ) \right)} ,
 \end{equation}
 where $e_1$ and $e_2$ are environment-related parameters. $\theta_{ku}$ is the elevation angle given by
 \begin{equation}
 	\theta_{ku}= \tfrac{180}{\pi} \arcsin \left(\tfrac{|h_u^a-h_k| }{||\bm b_u -\bm m_k|| } \right),  
 \end{equation}
 and $\theta_{ku} \in [0^\circ,90^\circ]$. Accordingly, the expected channel power gain is equal to
\begin{equation}\label{eq:gain}
    g_{ku}=\tfrac{\mathbb{P}(LoS,\theta_{ku})+(1-\mathbb{P}(LoS,\theta_{ku}))\kappa }{\gamma_0 ||\bm b_u -\bm m_k||^\alpha } =\tfrac{\hat{\mathbb{P}}( LoS,\theta_{ku}) }{\gamma_0 ||\bm b_u -\bm m_k||^\alpha }, 
\end{equation}
where $\kappa<1$ is the attenuation effect of the NLoS channel, $\alpha \geq 2$ is the path loss exponent and $\gamma_0=(\tfrac{4\pi f_c}{c} )^2$ is the reference free space path loss at a distance of 1 m with  $f_c$ being the carrier frequency of the G2A channel and $c$ being the speed of light. Then, the uplink transmission rate (bits/s) between the user $k$ and the UAV is 
\begin{equation}\label{eq:g2arate}
    R_{ku}=W_{ku}\log_2\left(1+\tfrac{P_kg_{ku}}{W_{ku} N_0}\right)=W_{ku}\log_2\left(1+\tfrac{\hat{\mathbb{P}}( LoS,\theta_{ku})P_k }{W_{ku} N_0\gamma_0 ||\bm b_u -\bm m_k||^\alpha}\right) , 
\end{equation}
where $W_{ku}$ is the communication bandwidth, $P_{k}$ is the transmission  power of user $k$, $N_0$ is the noise power spectral density and $\alpha$ is the path loss exponent of the G2A channel. 

Due to the rich scattering environment in the ground, the ground-to-ground (G2G) channel between the user device and the terrestrial BS can be modeled as a Rayleigh fading channel, which consists of large-scale attenuation and time-varying small-scale fading. We denote the small-scale fading coefficient as $h_{\text{G2G} }(t)$ and $ \mathbb{E}[|h_{\text{G2G} }(t)|^2 ]=1$. The instantaneous capacity of the G2G channel is given by
\begin{align}
	R_{kn}(t)=W_{kn} \log_2\left(1+\tfrac{|h_{\text{G2G} }(t)|^2 P_k}{W_{kn} N_0 \gamma_0 ||\bm b_n -\bm m_k||^\beta}\right),
\end{align}
which changes dynamically in the fading channel. Given an outage probability tolerance $\varepsilon$,  an achievable transmission rate (bits/s) of the G2G Rayleigh fading channel can be calculated from Jensen's inequality as
\begin{align} \label{eq:g2grate}
	R_{kn}=W_{kn} \log_2\left(1+\tfrac{F^{-1}(\varepsilon ) P_k}{W_{kn} N_0 \gamma_0 ||\bm b_n -\bm m_k||^\beta}\right), 
\end{align}
where $F(\cdot)$ is the cumulative distribution function (CDF) of the Rayleigh fading coefficient, which can be obtained from empirical  measurements\cite{zhan2017energy}. The power gain of the G2G channel can be derived as
\begin{align} \label{eq:g2ggain}
	g_{kn}=\tfrac{F^{-1}(\varepsilon ) }{  \gamma_0 ||\bm b_n -\bm m_k||^\beta}. 
\end{align}

\subsection{Localization Model}

We consider the OTDoA positioning method to localize the ground UEs. To reduce the computational complexity of UE, each UE estimates its 3D location from the reference signals sent by four ANs, including all the three ground BSs and a UAV.  We assume that the accurate locations of terrestrial BSs and UAVs are known, since UAVs can use altimeters and satellite systems to locate themselves. Each UE measures the time-of-arrival (ToA) of localization signals from the aerial and terrestrial ANs. The TDoA observations are the time difference between the ToA of a reference node and the ToAs of the remaining ANs. Suppose that all ANs are accurately synchronized in time and BS 1 is the reference node, device $k$ estimates its location $\bm m_k$ by solving the TDoA equations:
\begin{align}\label{eq:tdoa}
	 \tau_{n1}&= (||\bm b_n-\bm m_k||- ||\bm b_1-\bm m_k|| )/c ,\quad n \in \mathcal U \cup  \{2,3\} , \end{align}
where $c$ is the propagation speed of the signals and $\tau_{n1}$ is the TDoA between the $n$-th anchor and the reference node BS 1. There exist many linear and non-linear estimation algorithms to obtain $\bm m_k$ from the TDoA equations\cite{9636988}. In this paper, we are interested in the theoretical bound of localization accuracy regardless of the specific estimation algorithm.   

Without loss of generality, we analyze the performance of the OTDoA method by considering the narrowband positioning reference signal (NPRS) in the narrowband IoT (NB-IoT) system. NPRS is an orthogonal frequency division multiplexing (OFDM) based downlink reference signal occupying 180 kHz bandwidth (one LTE resource block).  The variance of the ToA measurement using $N_{sub}$ OFDM subframes is expressed as \cite{del2012achievable}
\begin{align}\label{eq:vartoa}
	\sigma_{ \text{ToA}}^2=\tfrac{T_s^2}{N_{sub}\cdot 8\pi^2\cdot \sum_{s\in \mathcal S} \sum_{i \in \mathcal N_s }p_i^2 i^2 } \cdot \tfrac{1}{\text{SNR}} =\tfrac{\psi}{\text{SNR}}  ,
\end{align}
where $T_s$ and SNR are the symbol duration and signal-to-noise
ratio of the received signal, respectively. In one localization period, we have $N_{sub}$ OFDM subframes and each subframe have a set $\mathcal{S}=\{1,2,...,S\}$ of symbols containing NPRS signal. In symbol $s$, the subset of subcarriers containing NPRS signal is denoted as $\mathcal N_s$ and $p_i \in [0,1]$ is the relative power weight of each subcarrier $i$. According to \eqref{eq:vartoa}, the variance of the ToA measurement from  UAV $u$ to UE $k$ is given by
\begin{align}
	\sigma_{uk}^2= \psi \tfrac{WN_0}{  g_{ku} P_u},
\end{align}
where $W$ is the NPRS bandwidth, $N_0$ is the noise power density, $P_u$ is the transmission power of the UAV, and $g_{ku}$ is the power gain of G2A channel.

In the G2G channel,  advanced signal processing algorithms are needed to eliminate the impacts of multi-path and NLoS propagation on the ToA measurement. The random estimation error is commonly modelled as an additional Gaussian noise term with variance $\sigma_{\text{ NLoS}}^2$. With the power gain in \eqref{eq:g2ggain},  we derive the variance of the ToA measurements for the G2G channel as 
\begin{align}\label{eq:toag2g}
	\sigma_{nk}^2\triangleq \psi \tfrac{WN_0}{  g_{kn} P_n}+ \sigma_{\text{ NLoS}}^2=\psi \tfrac{WN_0\gamma_0 ||\bm b_n -\bm m_k||^\beta}{ F^{-1}(\varepsilon )  P_n} + \sigma_{\text{ NLoS}}^2 ,\quad n\in \mathcal{N} ,
\end{align}
where $P_n$ is the transmission power of the BS $n$. Then, the covariance matrix $\mathbf{R}_{\text{TDoA}}$ of the TDoA observations in unit $\text{s}^2$ is given by \cite{wang2019energy}
\begin{align} \label{eq:cov}
		\mathbf{R}_{\text{TDoA}}=\begin{bmatrix}
\text{var}(\tau_{21}) & \text{cov}(\tau_{21}, \tau_{31}) & \text{cov}(\tau_{21}, \tau_{u1})\\
\text{cov}(\tau_{31}, \tau_{21}) & \text{var}(\tau_{31}) & \text{cov}(\tau_{31}, \tau_{u1})\\
\text{cov}(\tau_{u1}, \tau_{21}) & \text{cov}(\tau_{u1}, \tau_{31})& \text{var}(\tau_{u1})
\end{bmatrix}
=\begin{bmatrix}
\sigma_{1k}^2+\sigma_{2k}^2 & \sigma_{1k}^2 & \sigma_{1k}^2\\
\sigma_{1k}^2 & \sigma_{1k}^2+\sigma_{3k}^2 & \sigma_{1k}^2\\
 \sigma_{1k}^2 & \sigma_{1k}^2& \sigma_{1k}^2+\sigma_{uk}^2
\end{bmatrix},
\end{align}
where var$(x)$ is the variance of variable $x$ and cov$(x,y)$ is the covariance of variables $x$ and $y$.

\subsection{Localization Performance Metrics}
The Fisher information matrix (FIM) quantifies the amount of information about the unknown parameter that measurement vector carries. The FIM for OTDoA positioning measurements in \eqref{eq:tdoa} is defined as 
\begin{align}
	\mathbf{F}=\mathbf{H} ^T \mathbf{R}^{-1}\mathbf{H},
\end{align}
where $\mathbf{R}=c^2\cdot \mathbf{R}_{\text{TDoA}}$ is the covariance matrix of OTDoA positioning in unit $\text{m}^2$, and   $ \mathbf{H}$ is the Jacobian matrix of the TDoA equations
\begin{align}
		\mathbf{H}=\begin{bmatrix}
\tfrac{\partial \tau_{21}}{\partial x_k} & \tfrac{\partial \tau_{21}}{\partial y_k} & \tfrac{\partial \tau_{21}}{\partial h_k}\\
\tfrac{\partial \tau_{31}}{\partial x_k} & \tfrac{\partial \tau_{31}}{\partial y_k} & \tfrac{\partial \tau_{31}}{\partial h_k}\\
\tfrac{\partial \tau_{u1}}{\partial x_k} & \tfrac{\partial \tau_{u1}}{\partial y_k} & \tfrac{\partial \tau_{u1}}{\partial h_k}
\end{bmatrix},
\end{align}
and the elements in the first and second rows of $\mathbf{H}$ can be derived as 
\begin{align} \label{eq:helement}
	\tfrac{\partial \tau_{n1}}{\partial x_k} = \tfrac{x_n^b- x_k}{||\bm b_n-\bm m_k||}- \tfrac{ x_1^b -x_k}{ ||\bm b_1-\bm m_k||},~
	\tfrac{\partial \tau_{n1}}{\partial y_k} = \tfrac{y_n^b- y_k}{||\bm b_n-\bm m_k||}- \tfrac{ y_1^b -y_k}{ ||\bm b_1-\bm m_k||},~
	\tfrac{\partial \tau_{n1}}{\partial h_k} = \tfrac{h_n^b- h_k}{||\bm b_n-\bm m_k||}- \tfrac{ h_1^b -h_k}{ ||\bm b_1-\bm m_k||},
\end{align}
with $n=2,3$. The elements in the third row of $\mathbf{H}$ can be obtained  from \eqref{eq:helement} by replacing $\bm b_n$ with $\bm b_u$. There are several localization performance metrics defined based on the FIM. A common metric of localization accuracy is CRLB, which represents the theoretical lower bound of mean square estimation error. CRLB can be expressed as  a function of UE location and UAV location $L(\bm m_k, \bm b_u )= \text{Tr} \left( \mathbf{F}^{-1} \right)$, where Tr($\cdot$) means the trace of matrix. However, it is difficult to derive an analytical expression of the feasible UAV location with CRLB. In this paper, we adopt the \textit{D-optimality} criterion to analyze the localization performance, which is a well-known performance metric of localization accuracy \cite{ucinski2004optimal}. With the \textit{D-optimality} metric, we will be able to find intuitive and meaningful geometrical characterizations of UAV deployment in Section \ref{sec:analysis}.  Specifically, \textit{D-optimality} seeks to maximize the following opt-$D$ value 
 \begin{align}
 	\text{opt-}D=\det (\mathbf{F}^{-1})= \tfrac{\det(\mathbf{H})^2}{\det(\mathbf{R})}. 
 \end{align}

The determinant of $\mathbf{F}^{-1}$ is inversely proportional to the uncertainty area of an unbiased location estimation \cite{zheng2018oparray}.  Intuitively, a \textit{D-optimality} design minimizes the uncertainty ellipsoid, or equivalently, minimize the average location estimation error.

\section{Problem Formulation}\label{sec:problem}

Before formulating the UAV deployment optimization problem, we first provide a motivating example of using UAV to assist the ground localization service. In particular, we show that the proposed UAV-assisted localization method can effectively  improve the localization accuracy of the conventional scheme using only terrestrial BSs. For the ease of exposition, we divide an area $\mathcal A= [0,600] ~\text{m}\times[0,600]~\text{m}$   into small boxes with width of 10 m each. Three BSs are deployed at location $(0,0,30)$, $(500,0,30)$ and $(250,433,30)$ with unit meter. At each candidate location, we optimize the horizontal placement of the fourth AN at $h=30$ m and calculate the CRLB of the localization variance. For comparison, we place the UAV as the fourth AN at a fixed amplitude 200 m and similarly optimize its horizontal location.\footnote{Detailed simulation parameters are provided in the simulation section. }

\begin{figure}[ht]
	\centering
	\subfigure[Horizontal variance without UAV ]{\includegraphics[width=0.45\textwidth ]{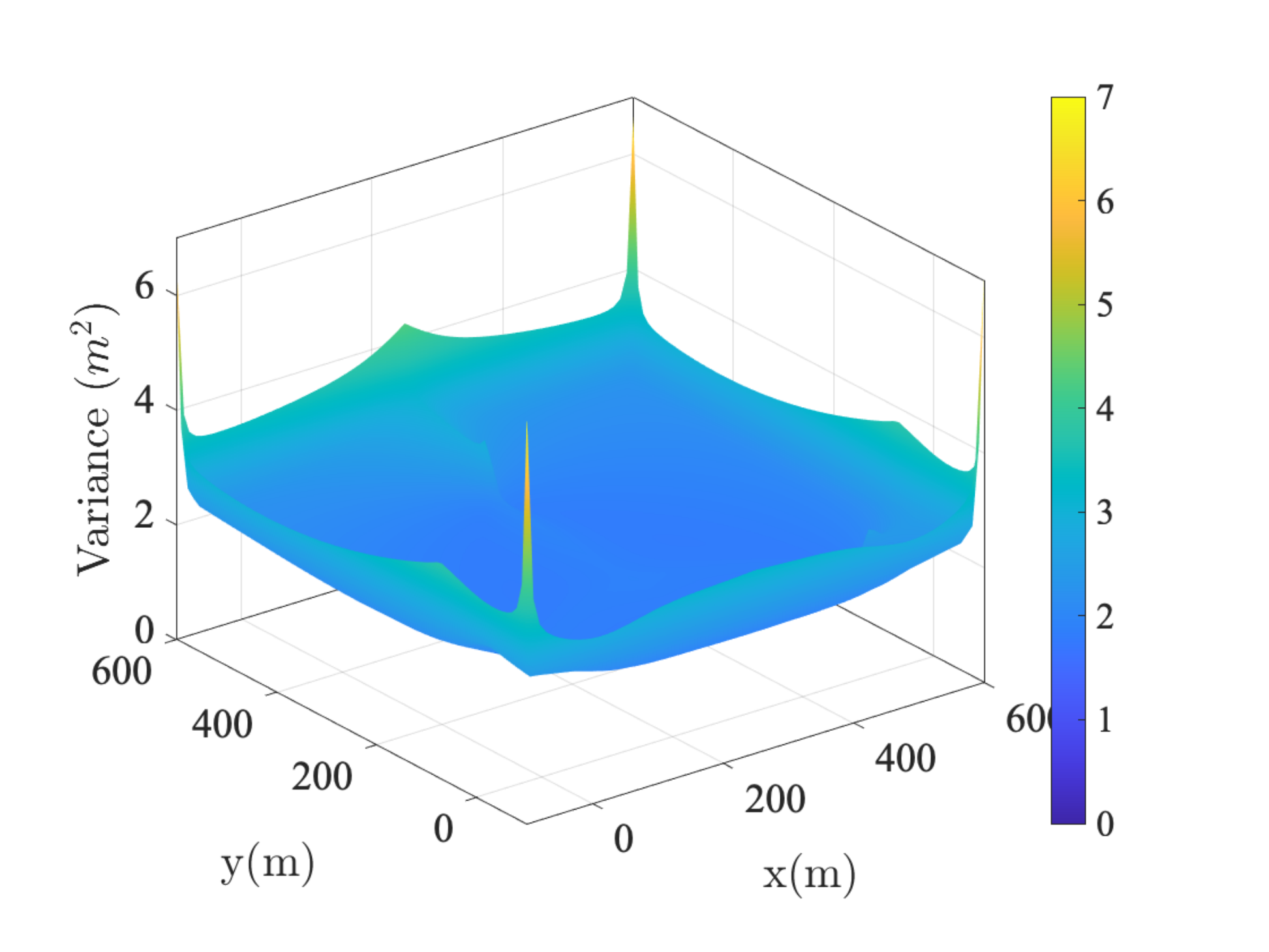}}
	\subfigure[Horizontal variance with UAV ]{\includegraphics[width=0.45\textwidth ]{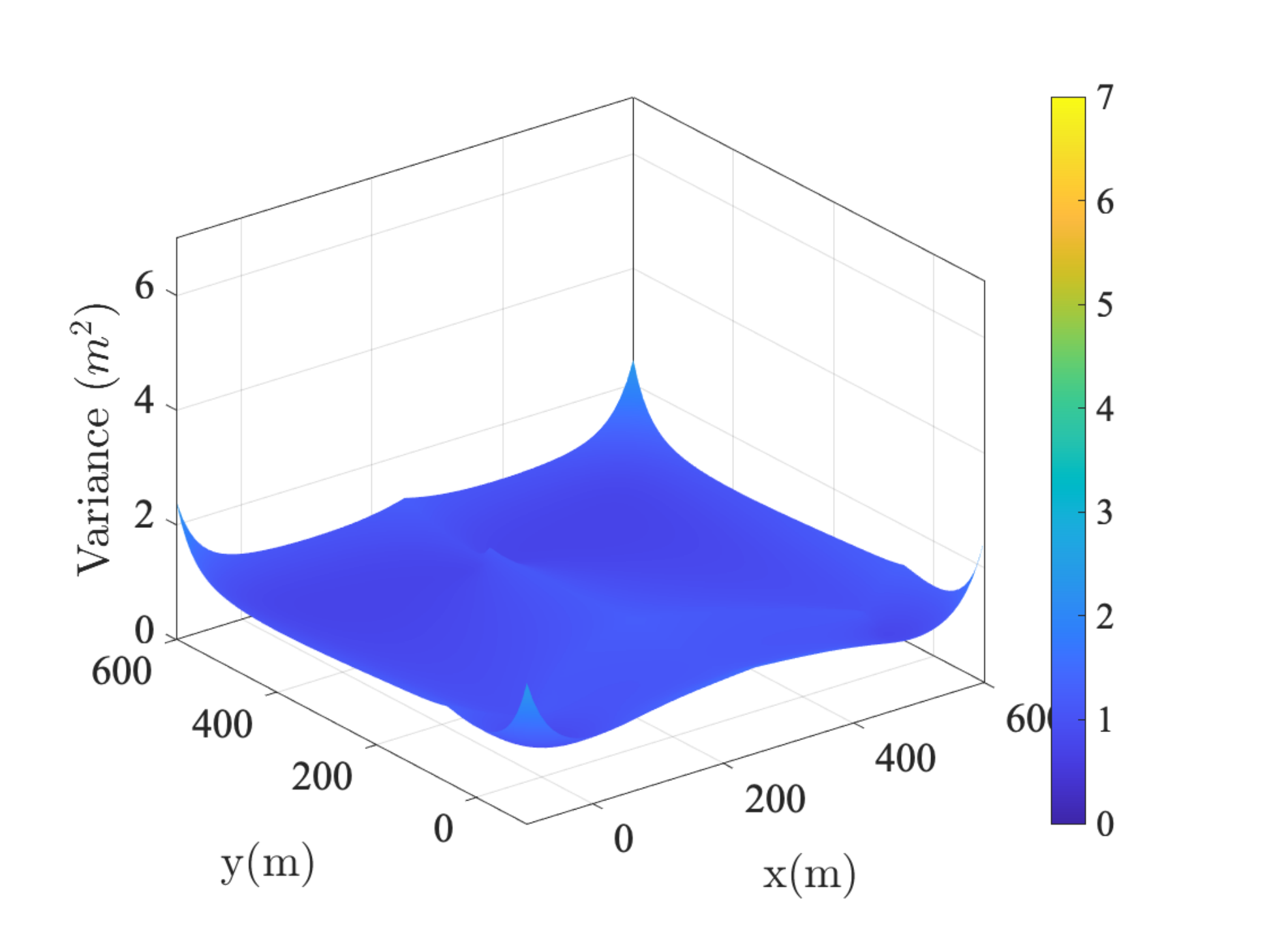}}
	\subfigure[Vertical variance without UAV ]{\includegraphics[width=0.45\textwidth ]{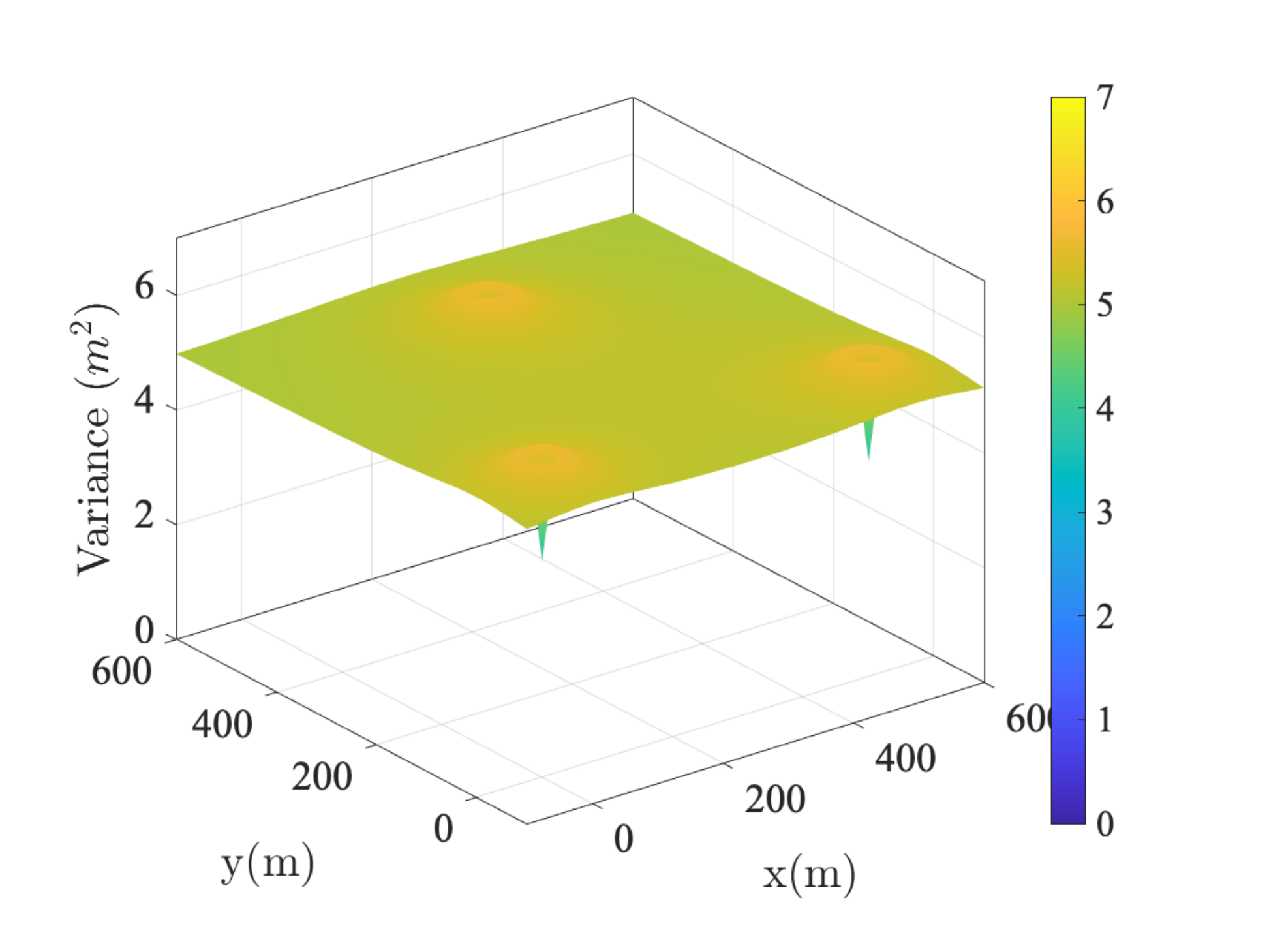}}
	\subfigure[Vertical variance with UAV]{\includegraphics[width=0.45\textwidth ]{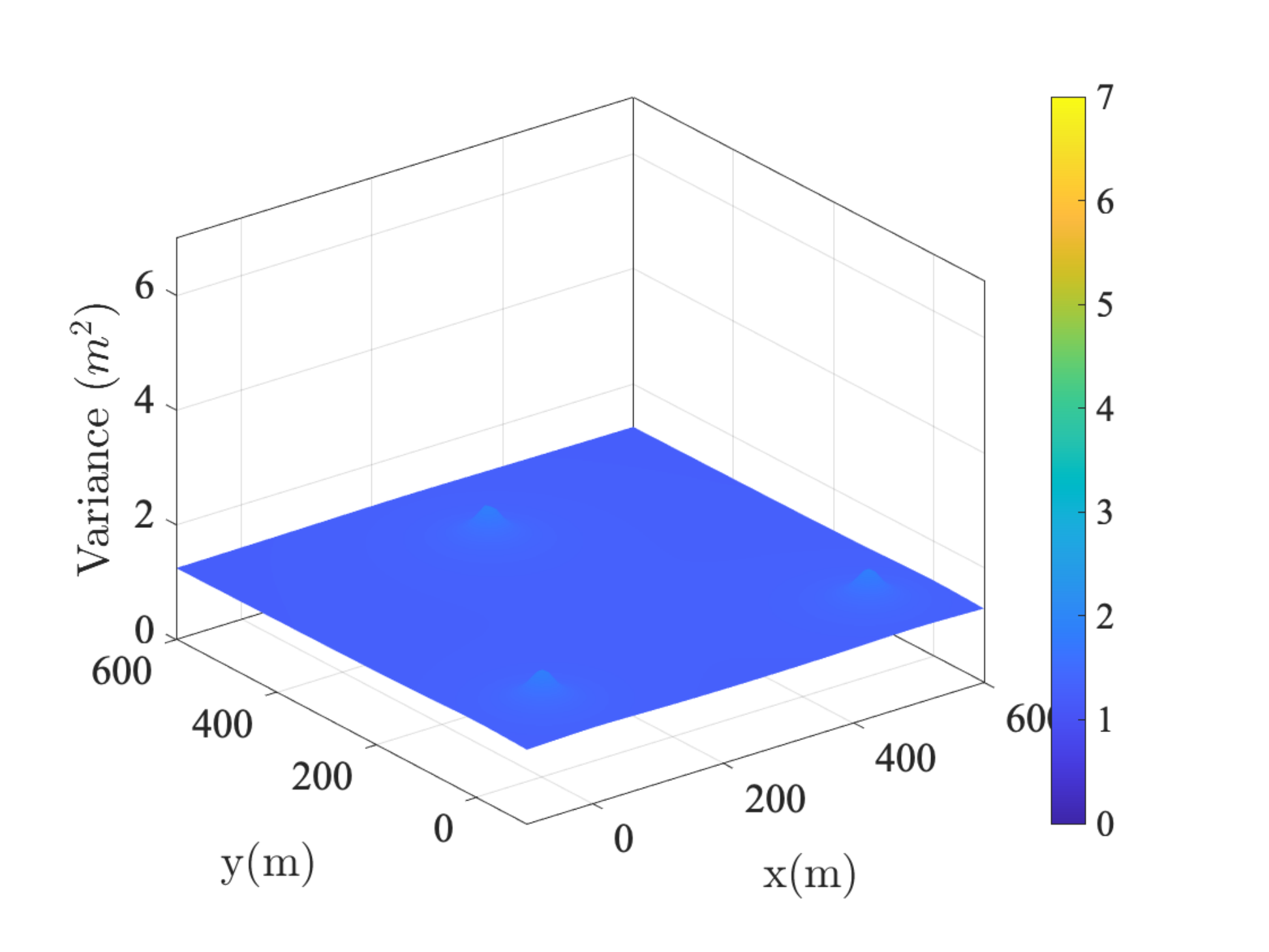}}
	\caption{Localization estimation variances under proposed design and conventional design, $\sigma^2_{NLoS}=4\times 10^{-17} $ $\text{s}^2$ }
	\label{fig:nloscomp}
\end{figure}

After traversing all the points to be localized, we generate the heap maps in Fig. \ref{fig:nloscomp} according to the recorded estimation variance, when $\sigma^2_{NLoS}$ is equal to $4\times 10^{-17}~\text{s}^2$. Each point in the heat map shows the minimum localization variance achievable by optimizing the location of the fourth AN, either the ground BS or UAV.   Fig. \ref{fig:nloscomp}(a) and Fig. \ref{fig:nloscomp}(b) show the horizontal variance under both designs. Without the UAV, the horizontal variance is in the range of $[1.80,7.12]$. With UAV as the additional AN, the horizontal variance reduced to the range of $[0.71,2.45]$. The improvement is more significant in the vertical direction as shown in   Fig. \ref{fig:nloscomp}(c) and Fig. \ref{fig:nloscomp}(d). The variance ranges are $[3.75,5.63]$ and $[1.25,1.88]$ for the conventional design and the UAV-assisted design, respectively.

\begin{figure}[ht]
	\centering
	\subfigure[Average localization variance versus $\sigma_{NLoS}^2$] {\includegraphics[width=0.45\textwidth ]{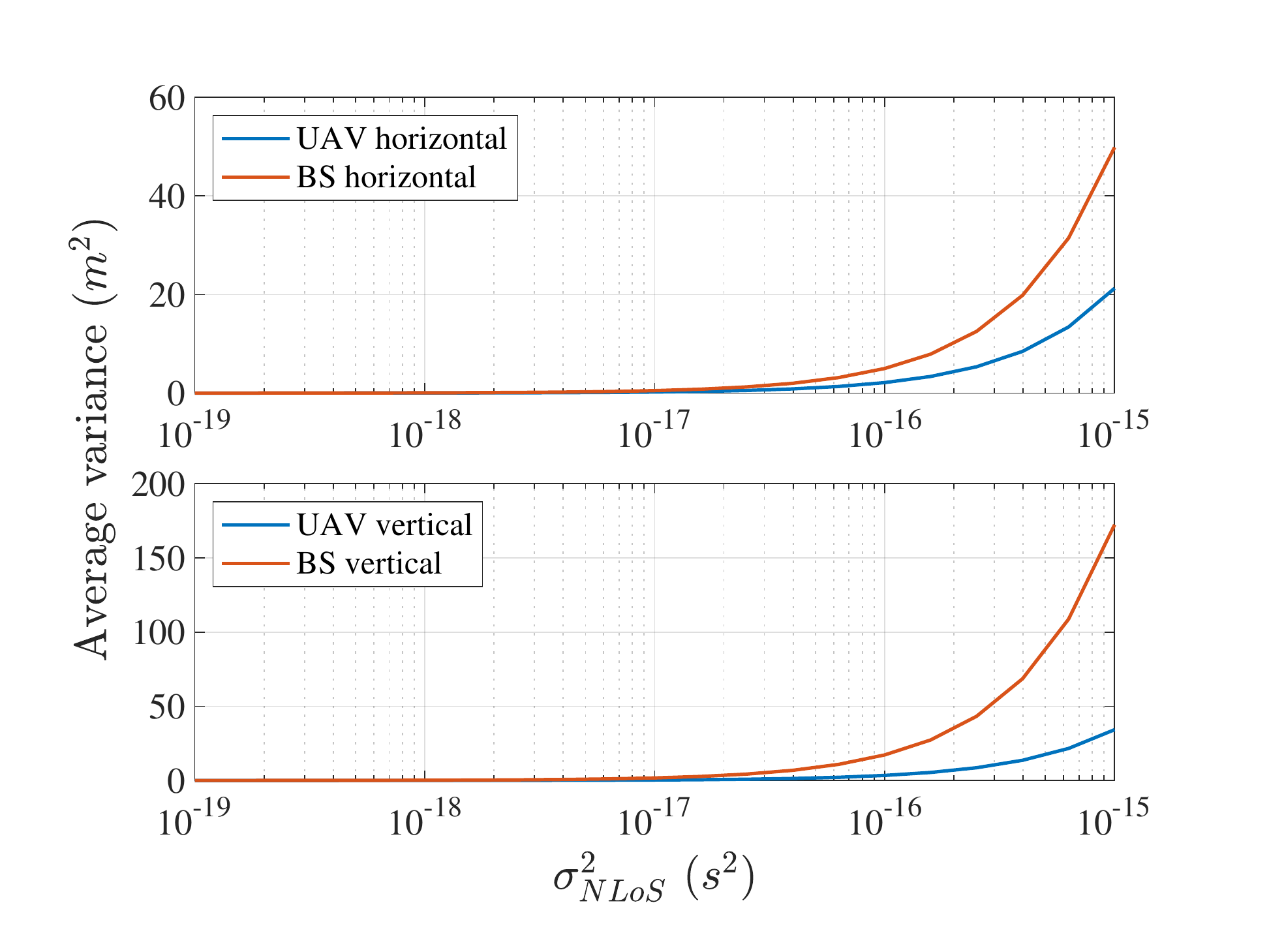}}
	\subfigure[Localization and communication capabilities]{
	\includegraphics[width=0.4\textwidth]{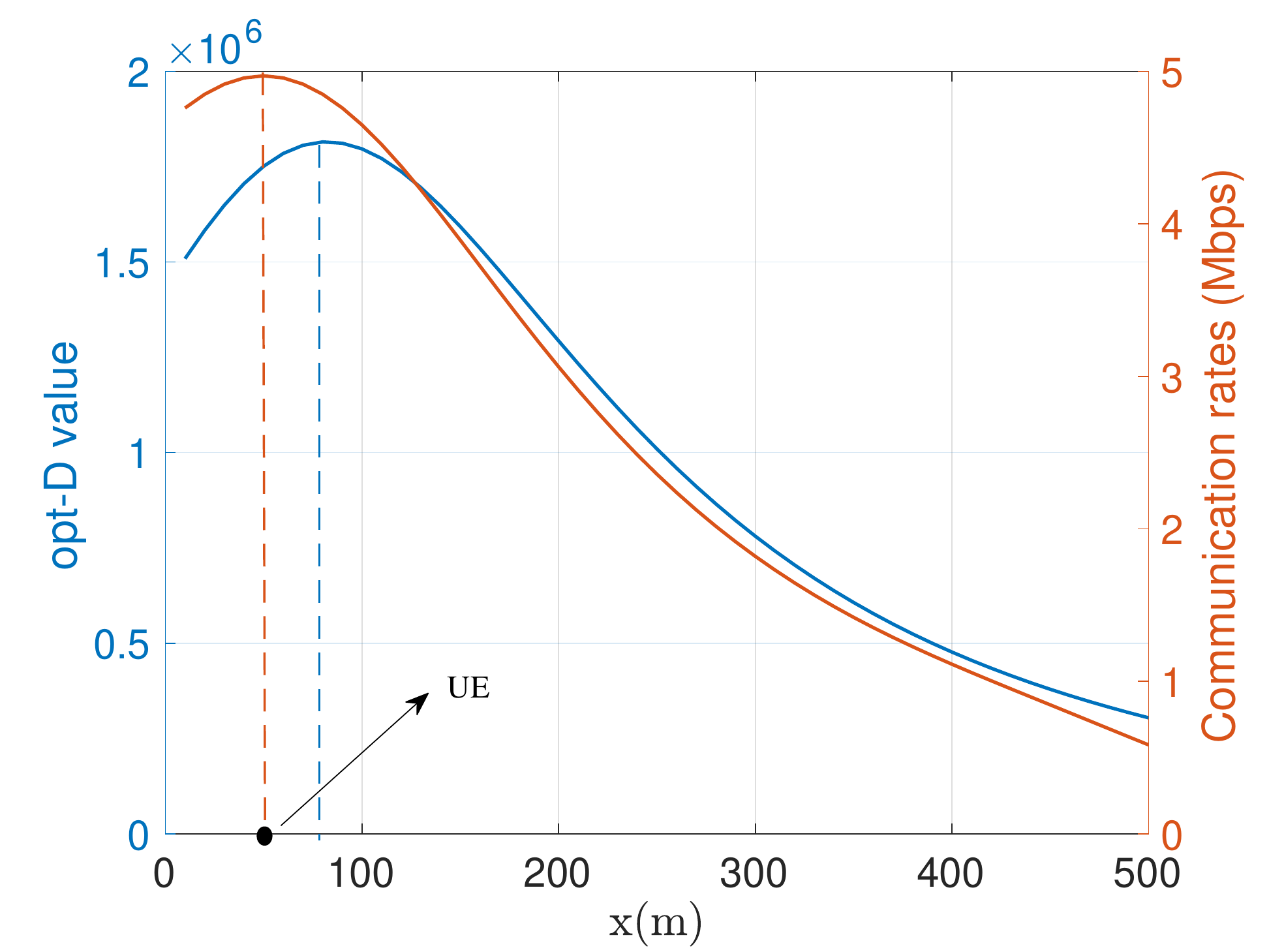}}
	\caption{Variance and capability comparisons}
	\label{fig:nlos}
\end{figure}

In Fig. \ref{fig:nlos}, we observe the localization accuracy improvement and the conflicting impact of UAV placement.  Fig. \ref{fig:nlos}(a) shows the average variance under different NLoS noise. Each sample point is the average performance of 10 localization points distributed in the area with random height from the range $[10,20]$ m. For both methods, the average variance will increase when the NLoS noise becomes larger. The UAV-assisted localization method produces much lower localization variance than the conventional scheme in both horizontal and vertical directions. The improvement is more significant in the vertical direction when $\sigma_{NLoS}^2$ is large. For example, when $\sigma_{NLoS}^2$ is equal to $10^{-15}~\text{s}^2$, the proposed method reduces $57.3\%$ horizontal variance and $80.1\%$ vertical variance.  To achieve localization accuracy at 1 m precision in both horizontal and vertical directions, the $\sigma_{NLoS}^2$ is required to be lower than $4\times 10^{-18}~\text{s}^2$ for the conventional method and $2.5\times 10^{-17}~\text{s}^2$ for the UAV-assisted method. The above  observations demonstrate the advantage of integrating UAVs for accurate 3D localization, especially in the vertical dimension. That is, the amplitude of UAV AN provides not only strong LoS link quality, but also the critical diversity in AN locations.

With fixed terrestrial BSs, we aim to deploy a minimum number of UAVs to satisfy the localization and communication requirements of all UEs. The problem is formulated as 
\begin{subequations}\label{eq:p1}
\begin{align}
	\mathcal{P}1: \quad& \min_{\{\bm b_u \}_{u\in \mathcal U} } \quad |\mathcal U |,\\
	\text{s.t.}\quad & \max_{\{\bm b_u \}_{u\in \mathcal U}} ~  \text{opt-}D(\bm m_k,\bm b_u ) \geq \epsilon_k, \forall k \in \mathcal K,\\
	& \max_{n\in \mathcal {U}\cup\mathcal{N} } R_{kn} \geq R_k^{th}, \forall k \in \mathcal K,\\
	& \bm b_u \in \mathcal A, \forall u\in \mathcal U, 
\end{align}
\end{subequations}
where $|\mathcal U |=U$ is the number of UAVs and $\mathcal A$ is the restricted area to deploy the UAVs. Constraint (\ref{eq:p1}b) is the localization accuracy requirement specified by an individual threshold $\epsilon_k$. Constraint (\ref{eq:p1}c) means that the communication rate of UE $k$ should be larger than the threshold $R_k^{th}$. The maximization term in (\ref{eq:p1}b) means that each UE estimates its location using all the three ground BSs and only one UAV. Meanwhile, the maximization term in (\ref{eq:p1}c) means that each UE transmits information to only one base station, either a UAV or a ground BS. 

The optimization problem is a cardinality minimization, which is NP-hard in general. An intuitive methodology is to search feasible solutions given fixed cardinality. However,  this does not reduce the difficulty because the localization constraints are non-convex in regard to the UAV positions. In addition, $\mathcal P1$ contains two implicit association problems between UEs and BSs or UAVs, which are hard combinatorial optimization problems in multi-UAV and multi-BS scenarios.  Another major difficulty arises from the potential conflicting impact of UAV placement to the communication and localization performance. As an illustrating example, in Fig. \ref{fig:nlos}(b), we constraint the UAV flying at a fixed amplitude and in a straight line passing over a UE at $x=50$ m. We see that from $x=50$ to $80$ m, the communication rate drops while the localization accuracy improves.

 Overall,  $\mathcal P1$ is difficult to solve, and it is even difficult to determine whether there exists a feasible solution satisfying all the constraints given the number of UAVs. In the next section, we analyze the localization performance constraints (\ref{eq:p1}b) and derive closed-form expressions of  feasible UAV hovering region.

\section{Geometric Characterization of Localization Constraints} \label{sec:analysis}
  In this section, we characterize the localization performance constraint (\ref{eq:p1}b) in a geometric form. This enables the design of efficient graph-based algorithm to solve $\mathcal{P}1$ in the next section. 

\subsection{Approximation of opt-$D$ Value}
To facilitate the analysis,  we propose in this subsection an approximation of the opt-$D$ value for a tagged UE $k$. For simplicity, we drop the index $k$ and denote $\sigma_{nk}$ as $\sigma_n$. We calculate the determinant of the covariance matrix
 \begin{align} \label{eq:optd}
 	\det(\mathbf{R})=c^6 \sigma_1^2\sigma_2^2\sigma_3^2+c^6 (\sigma_1^2\sigma_3^2+\sigma_1^2\sigma_2^2+\sigma_2^2\sigma_3^2)\sigma_u^2 \triangleq D_1+D_2 \cdot \tfrac{ \psi}{\text{SNR}_u},  
 \end{align}
where  $D_1\triangleq c^6 \sigma_1^2\sigma_2^2\sigma_3^2, D_2\triangleq c^6( \sigma_1^2\sigma_3^2+\sigma_1^2\sigma_2^2+\sigma_2^2\sigma_3^2)$. Supposing that $D_1 \gg D_2 \cdot \tfrac{ \psi}{\text{SNR}_u}$, we can ignore the second term and  approximate $\det(\mathbf{R})$ using $D_1$. Then, we denote the corresponding value as opt-$D_1$ 
 \begin{align}\label{eq:optd1}
 	\text{opt-}D_1=\tfrac{\det(\mathbf{H})^2}{D_1}.
 \end{align}
 In contrast, when $D_1 \ll D_2 \cdot \tfrac{ \psi}{\text{SNR}_u}$, we ignore the constant term $D_1$ and approximate $\det(\mathbf{R})$ using the second term. We denote this approximation as opt-$D_2$, where 
  \begin{align}
 	\text{opt-}D_2=\tfrac{\det(\mathbf{H})^2 \psi }{D_2\text{SNR}_u }.
 \end{align}
 In the following lemma, we show the condition when  opt-$D_1$ is an accurate approximation.
\begin{lemma}
  When $\frac{\text{SNR}_u}{\text{SNR}_n}\gg 3-\tfrac{ \sigma_{\text{ NLoS}}^2}{\sigma_u^2}, \forall n\in \mathcal{N}$,   opt-$D$  converges to  opt-$D_1$. 
\end{lemma}
\begin{proof}[Proof:\nopunct] The condition implies
\begin{align}
& \tfrac{\text{SNR}_u}{\text{SNR}_n}\gg 3-\tfrac{ \sigma_{\text{ NLoS}}^2\text{SNR}_u}{\psi} \Rightarrow  \tfrac{\text{SNR}_u}{\text{SNR}_n}+\tfrac{ \sigma_{\text{ NLoS}}^2\text{SNR}_u}{\psi}\gg 3 \notag \\
\Rightarrow & \tfrac{\psi}{\text{SNR}_n}+\sigma_{\text{ NLoS}}^2  \gg \tfrac{3\psi}{\text{SNR}_u} \Rightarrow \tfrac{\sigma_n^2}{3}  \gg \sigma_u^2,~\forall n \in N \notag \\
    \Rightarrow  & \tfrac{\sigma_1^2\sigma_2^2\sigma_3^2}{3}  \gg \sigma_2^2\sigma_3^2\sigma_u^2,\tfrac{\sigma_1^2\sigma_2^2\sigma_3^2}{3}  \gg \sigma_1^2\sigma_3^2\sigma_u^2,\tfrac{\sigma_1^2\sigma_2^2\sigma_3^2}{3}  \gg \sigma_1^2\sigma_2^2\sigma_u^2.  
\end{align}
Summing up both sides of the three inequalities, we have
\begin{align}
\sigma_1^2\sigma_2^2\sigma_3^2 \gg (\sigma_1^2\sigma_3^2+\sigma_1^2\sigma_2^2+\sigma_2^2\sigma_3^2)\sigma_u^2 
  \Leftrightarrow  D_1 \gg D_2 \cdot \tfrac{ \psi}{\text{SNR}_u},
\end{align}
which is a sufficient condition for opt-$D$ converging to opt-$D_1$.
\end{proof}

 In practice, the random estimation error $\sigma_{\text{ NLoS}}^2$ could be larger than the noise variance of G2A channel $\sigma_u^2$,  and thus the condition for Lemma 1 is easily satisfied. In Fig. \ref{fig:dvalues}(a), we compare the two approximations and the opt-$D$ value with different ratios $\frac{\text{SNR}_u}{\text{SNR}_n} \in [10^{-2}, 10] $, and depict the values in a log scale. We observe that opt-$D_1$ is much more accurate than opt-$D_2$ within the considered ratio range. In Fig. \ref{fig:dvalues}(b), the opt-$D$ and opt-$D_1$  values are almost on top with each other when the ratio is greater than 1, and the gap is less that 2\% when the ratio is larger than 0.1. Therefore, we can safely adopt opt-$D_1$ as an accurate approximation in the following analysis. 
 
\begin{figure}[ht]
	\centering
	\subfigure[The logarithm of $\det(\mathbf{R})$ versus the SNR ratio]{ \includegraphics[width=0.45\textwidth ]{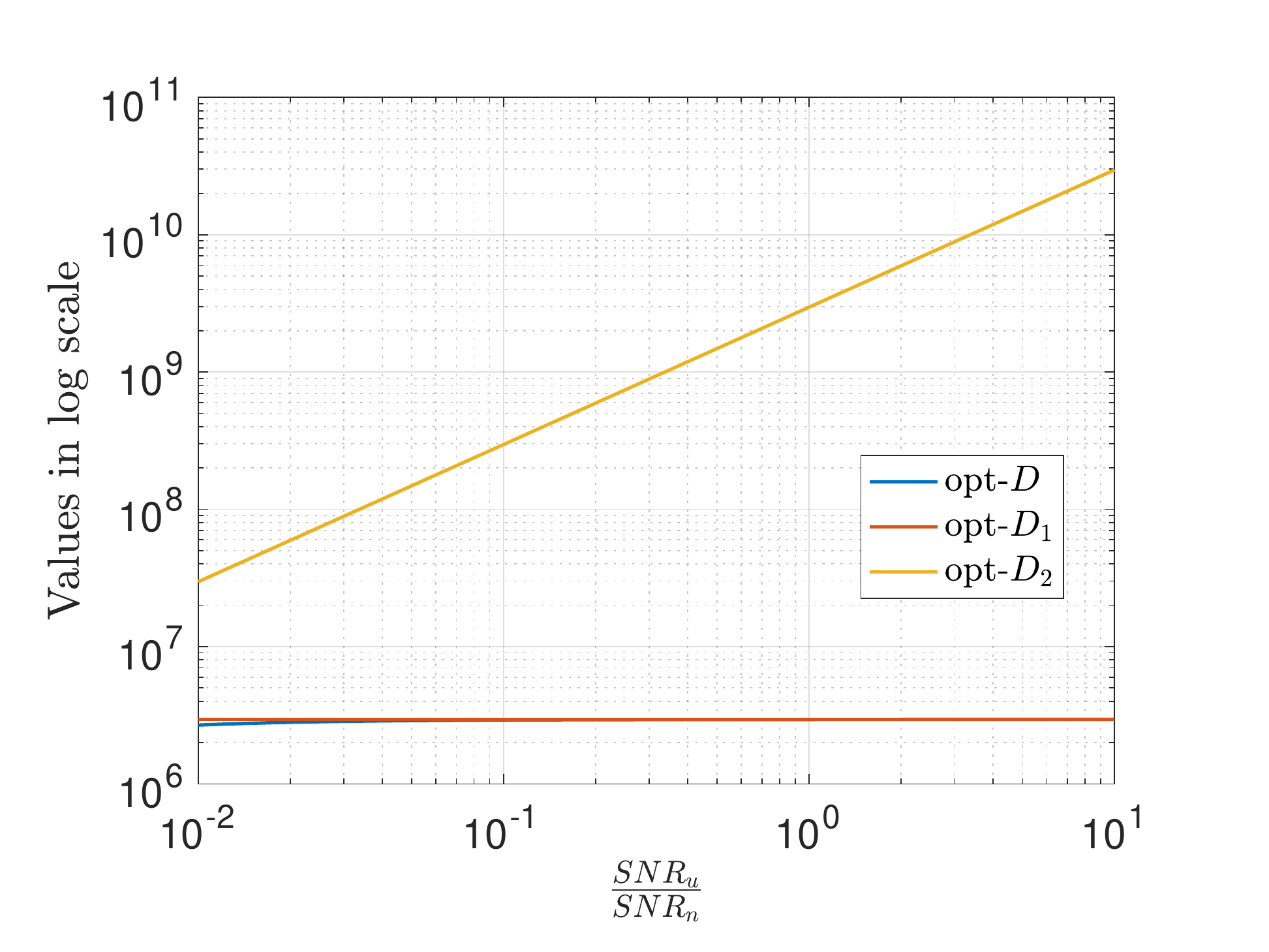}}
	\subfigure[opt-$D_1$  and opt-$D$ ]{\includegraphics[width=0.45\textwidth ]{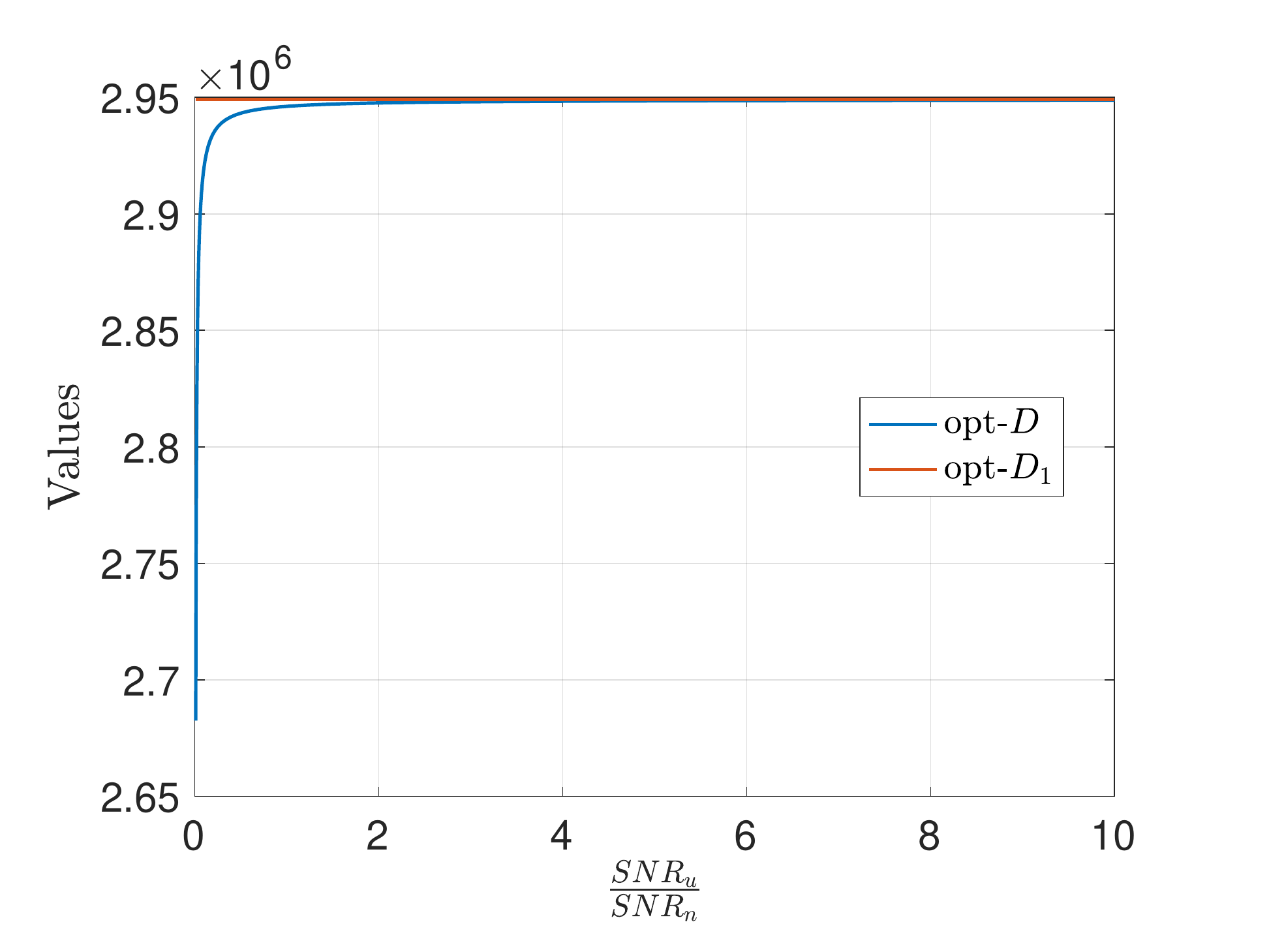}}
	\caption{Approximations of opt-$D$ values}
	\label{fig:dvalues}
\end{figure}

\subsection{Equivalent Geometric Characterization}
 Firstly, we define the unit vector from the $u$-th UAV to a location point $ \bm m=[x_0,y_0,h_0]$ as $ \bm q_u$, and that from ground AN $n$ to $\bm m$ as $ \bm q_n$, where
\begin{align}
    \bm q_u &= \tfrac{\bm{b_u-m}}{||\bm{b_u-m}||}=[q_{u1},q_{u2},q_{u3}], \forall u \in \mathcal{U},\label{eq:qu} \\
    \bm q_n &= \tfrac{\bm{b_n-m}}{||\bm{b_n-m}||}=[q_{n1},q_{n2},q_{n3}], \forall n\in \mathcal{N}. \label{eq:qn}
\end{align}

Then, $\mathbf {H}$ can be represented by the unit vectors
\begin{align}
	\mathbf {H}=\begin{bmatrix}
		\bm q_2-\bm q_1\\
		\bm q_3-\bm q_1\\
		\bm q_u-\bm q_1
	\end{bmatrix},
\end{align}
and we explicitly express $\det(\mathbf{H})$ using the elements in \eqref{eq:helement} such that
 \begin{align}\label{eq:detH}
 	\det (\mathbf{H})= \alpha_1\left(\tfrac{x_u^a-x_0}{||\bm b_u-\bm m||}- q_{11} \right)+ \alpha_2\left(\tfrac{y_u^a-y_0}{||\bm b_u-\bm m ||}- q_{12} \right)+ \alpha_3\left(\tfrac{h_u^a-h_0}{||\bm b_u-\bm m ||}- q_{13} \right),
 \end{align}
 where the coefficients $\alpha_n$ for $n=1,2,3$ are derived using \eqref{eq:qn}
 \begin{align}
 	\alpha_1&=(q_{22}-q_{12})(q_{33}-q_{13})-(q_{23}-q_{13})(q_{32}-q_{12}),\\
 	\alpha_2&=(q_{23}-q_{13})(q_{31}-q_{11})-(q_{21}-q_{11})(q_{33}-q_{13}),\\
 	\alpha_3&=(q_{21}-q_{11})(q_{32}-q_{12})-(q_{22}-q_{12})(q_{31}-q_{11}).
 \end{align}

The localization requirement $ \text{opt-}D \geq \epsilon$ with opt-$D_1$  approximation in \eqref{eq:optd1} becomes
\begin{align}\label{eq:locreq}
\tfrac{\det(\mathbf{H})^2}{D_1}	\geq \epsilon.
\end{align}

Based on the value of $\det(\mathbf{H})$, we derive the implicit geometry of \eqref{eq:locreq} in two cases. 
\subsubsection{Case 1} 

When $ \det(\mathbf{H})\geq 0 $, constraint \eqref{eq:locreq} is equivalent to 
\begin{align}\label{eq:d1ineq}
 	 \alpha_1(x_u^a-x_0) +\alpha_2(y_u^a-y_0)+\alpha_3(h_u^a-h_0)    \geq \tilde \epsilon_1 ||\bm b_u-\bm m ||,
 \end{align}
 where $\tilde \epsilon_1=\sqrt{D_1 \epsilon } +\alpha_1 q_{11}+\alpha_2 q_{12}+\alpha_1 q_{13}$. The constraint \eqref{eq:d1ineq} is a standard second-order cone constraint of dimension 4, and it is equivalent to 
 \begin{align}\label{eq:secondcone}
     \alpha_1 q_{u1} +\alpha_2 q_{u2}+\alpha_3 q_{u3}  \geq \tilde \epsilon_1.
 \end{align}
 
 Geometrically, the relative UAV location $\bm q_u = [q_{u1},q_{u2},q_{u3}]$ that satisfies \eqref{eq:secondcone} can be viewed as the intersection of the sphere $\mathbb{S}=\{[x,y,h]: x^2+y^2+h^2=1 \} $ and the half space $\mathbb{H}=\{[x,y,h]: \alpha_1 x +\alpha_2 y+\alpha_3 h \geq \tilde \epsilon_1 \}$. The normal vector of the plane $\mathbb{P}=\{[x,y,h]: \alpha_1 x +\alpha_2 y+\alpha_3 h =\tilde \epsilon_1 \}$ is $\bm n=[\alpha_1,\alpha_2,\alpha_3]$ and the perpendicular distance from the origin to the plane is $\rho = \tfrac{|\tilde \epsilon_1 |}{\sqrt{\alpha_1^2 +\alpha_2^2+\alpha_3^2}} $.    The intersecting circle has radius $r=\sqrt{1-\rho^2}$ and its center is 
 \begin{align}
 	\bm c_0= \rho \tfrac{\bm n}{||\bm n ||}=\rho \tfrac{[\alpha_1,\alpha_2,\alpha_3]}{\sqrt{\alpha_1^2 +\alpha_2^2+\alpha_3^2}}= \tfrac{|\tilde \epsilon_1|}{\alpha_1^2 +\alpha_2^2+\alpha_3^2}[\alpha_1,\alpha_2,\alpha_3].
 \end{align}
 
For the sphere and plane to have an intersection, it must hold that $\rho \leq 1$, i.e., $ |\tilde \epsilon_1|\leq \sqrt{\alpha_1^2 +\alpha_2^2+\alpha_3^2}$. Accordingly, we establish in the following proposition a condition to check the feasibility of localization requirement.   
 \begin{proposition}\label{prop1}
When $ \det(\mathbf{H})\geq 0 $, the localization requirement for UE at $\bm m$ is feasible if and only if the value of $\epsilon$  satisfies
 	\begin{align}\label{eq:epsrange}
 			0 \leq \epsilon \leq \tfrac{ \left(c_1-c_2\right )^2}{D_1}, 
 	\end{align}
 	where $c_1$ and $c_2$ are given by
 	\begin{align}\label{eq:c1c2}
 		c_1=\sqrt{\alpha_1^2 +\alpha_2^2+\alpha_3^2},\quad
 		c_2=\alpha_1 q_{11}+\alpha_2 q_{12}+\alpha_3 q_{13}. 
 	\end{align}
 \end{proposition}
 \begin{proof}[Proof:\nopunct]
Please see the detailed proof in Appendix \ref{apx:prop1}.
\end{proof}

   \begin{figure}[ht]
	\centering
	\subfigure[{$\bm m=(250,135,10)$}]{\includegraphics[width=0.32\textwidth ]{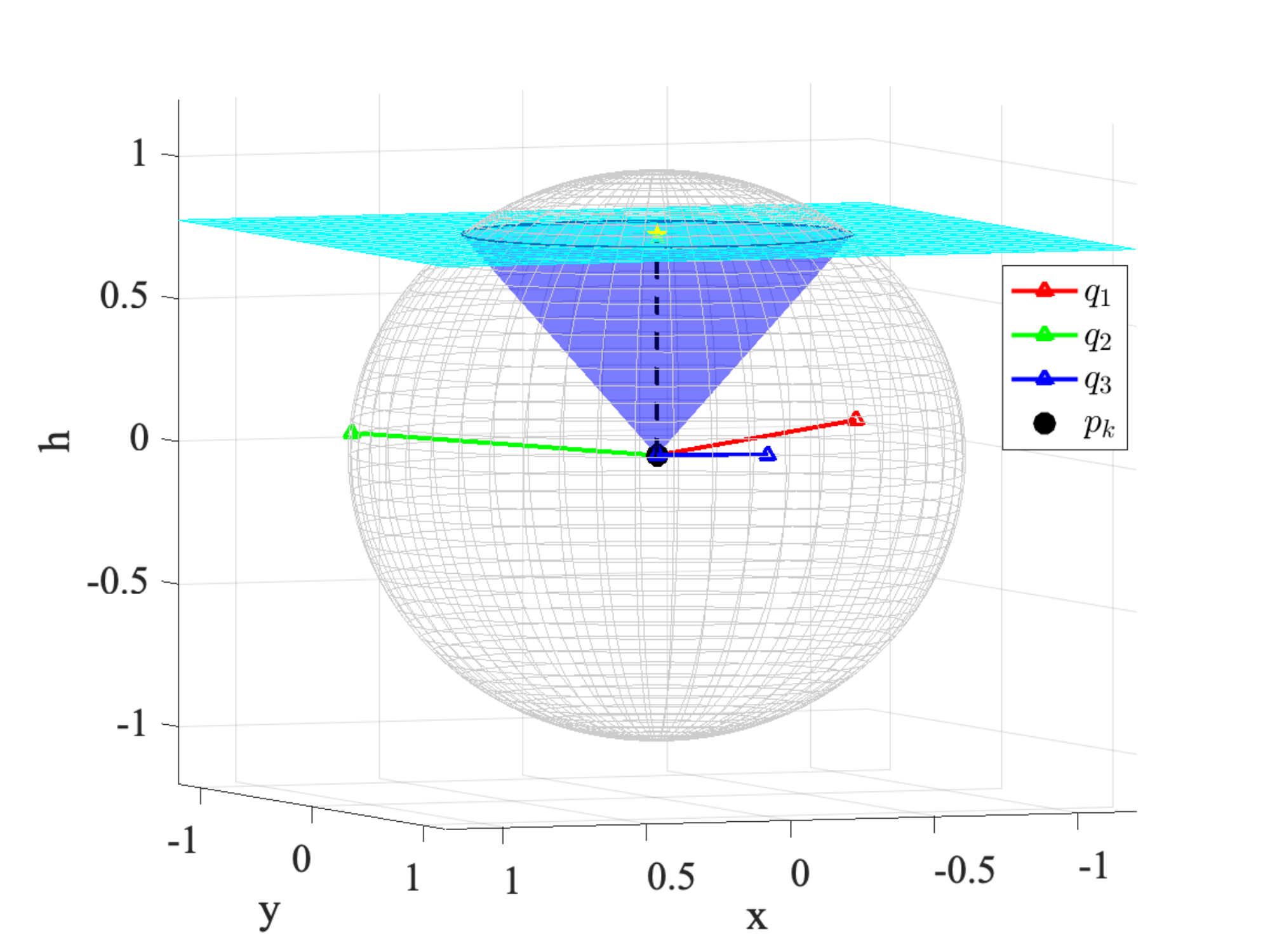}}
	\subfigure[{$\bm m=(100,50,10)$}  ]{\includegraphics[width=0.32\textwidth ]{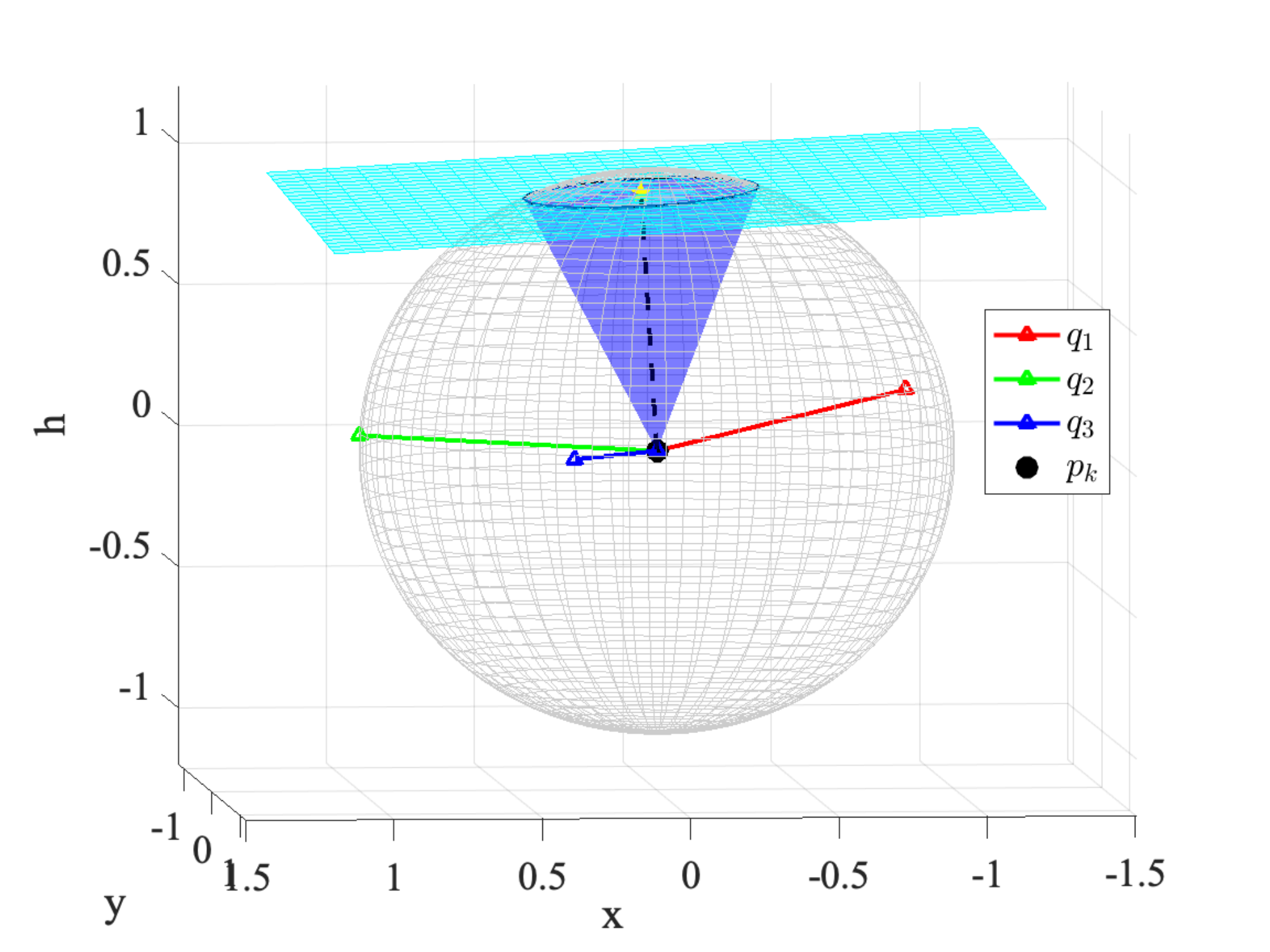}}
	\subfigure[{$\bm m=(0,10,10)$}]{\includegraphics[width=0.32\textwidth ]{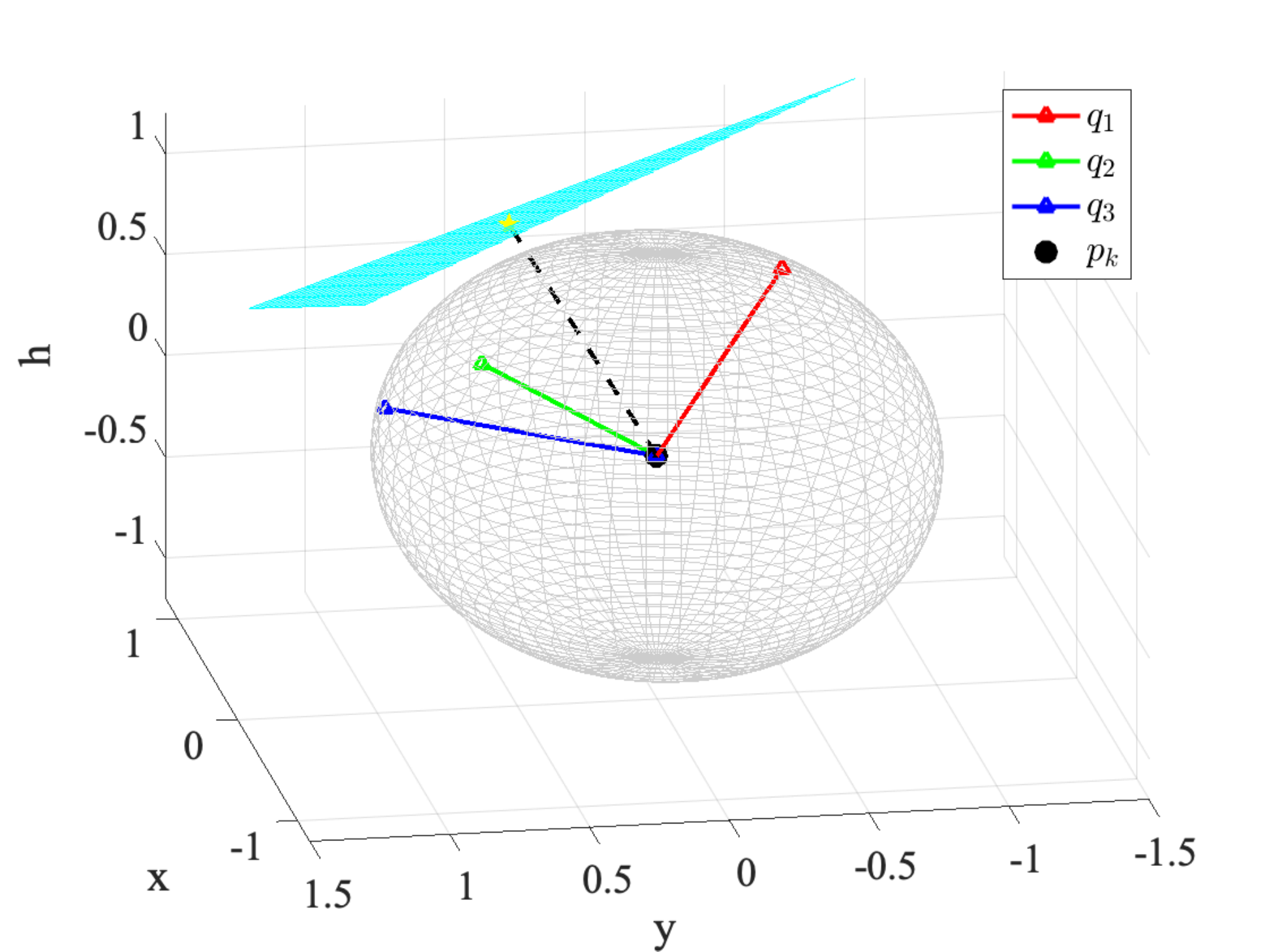}}
	\caption{Feasibility of different location points under $ \epsilon=8 \times 10^{-3}$}
	\label{fig:inter2}
\end{figure}

Proposition \ref{prop1} shows that by setting a reasonably small $\epsilon$ specified by \eqref{eq:epsrange}, we can always find a UAV location to satisfy the localization accuracy requirement. Otherwise, if $\epsilon$ is beyond the range, it could be the case that no UAV placement can achieve the required accuracy level. Fig. \ref{fig:inter2} showcases that the feasibility of different UE location points under the same settings. For UE location $(250,135,10)$, the localization requirement is feasible when $\epsilon \leq 0.12$; for UE location $(100,50,10)$, it is feasible when $\epsilon \leq 0.081$; for UE location $(0,10,10)$, it is feasible when $\epsilon \leq 7.7\times 10^{-3} $. As shown in Fig. \ref{fig:inter2}, under the same $\epsilon=8 \times 10^{-3}$,  the sphere and the plane has intersection for UE locations $(250,135,10)$ and $(100,50,10)$, which means that the localization requirement is feasible for some UAV deployment solution. In Fig. \ref{fig:inter2}(c), the sphere and the plane has no intersection, showing that no UAV deployment can satisfy the localization accuracy requirement for UE location $(0,10,10)$ when $\epsilon=8 \times 10^{-3}$. Proposition \ref{prop1}  provides a convenient method to set a proper $\epsilon$ for the constraint (\ref{eq:p1}b) to be feasible. Suppose that all the localization constraints (\ref{eq:p1}b) are feasible by Proposition \ref{prop1}. We continue to derive in the following proposition the feasible UAV hovering region for (\ref{eq:p1}b) assuming the UAV flies at a fixed altitude $h_f$.

  \begin{proposition}\label{prop2}
    For a fixed altitude $h_f$, the the localization constraint \eqref{eq:locreq} is satisfied if and only if the UAV $\bm b_u = [x^a_u, y^a_u,h_f]$ lies within  an ellipse defined by the following polynomial equation
  \begin{align}\label{eq:ellipse1}
 	\mathcal{E}^1:\quad &(\tfrac{\alpha_1^2 }{\tilde \epsilon_1^2}-1 )(x_u^a-x_0) ^2+\tfrac{2\alpha_1 \alpha_2 }{\tilde \epsilon_1^2} (x_u^a-x_0)(y_u^a-y_0)+(\tfrac{\alpha_2^2 }{\tilde \epsilon_1^2}-1 )(y_u^a-y_0)^2 \notag \\
 	&+\tfrac{2\alpha_1 \alpha_3 h_r }{\tilde \epsilon_1^2}(x_u^a-x_0) +\tfrac{2\alpha_2 \alpha_3 h_r }{\tilde \epsilon_1^2}(y_u^a-y_0)+ (\tfrac{\alpha_3^2}{\tilde \epsilon_1^2}-1 )h_r^2=0,
 \end{align}
 where $h_r=h_f-h_0$ is the relative altitude. The center $\bm c_1^e$ of the ellipse, the rotation angle $\theta_1$,  the semi-major axis $l_1^{a}$ and the semi-minor axis $l_1^{i}$ are respectively given by 
 \begin{align} \label{eq:center1}
 	\bm c_1^e=[\tfrac{\alpha_1 \alpha_3 h_r}{\tilde \epsilon_1^2-\alpha_1^2-\alpha_2^2}+x_0,\tfrac{\alpha_2 \alpha_3 h_r}{\tilde \epsilon_1^2-\alpha_1^2-\alpha_2^2}+y_0, h_f ],\quad \theta_1&=\arctan(-\tfrac{\alpha_1}{\alpha_2} ),\\
  	l_1^{a}= \sqrt{\tfrac{h_r^2(c_1^2-\tilde \epsilon_1^2 )}{\tilde \epsilon_1^2-\alpha_1^2-\alpha_2^2 }},\quad
  	l_1^{i}=\tfrac{h_r\tilde \epsilon_1 \sqrt{c_1^2-\tilde \epsilon_1^2 }}{\tilde \epsilon_1^2-\alpha_1^2-\alpha_2^2}. \label{eq:axes1}
  \end{align}
  \end{proposition}
\begin{proof}
 We substitute $h_u^a=h_f$ to let the equality in \eqref{eq:d1ineq} hold. Then, we square both sides of \eqref{eq:d1ineq} and reformulate it into \eqref{eq:ellipse1}. For an ellipse with the general form
  \begin{align}\label{eq:ge1}
      Ax^2 + Bxy + Cy^2 + Dx + Ey + F = 0,
  \end{align}
  its center $\bm c^e$, rotation angle $\theta^r$, semi-major axis $l^{a}$ and semi-minor axis $l^{i}$ are given by 
   \begin{align} 
 	\bm c^e=[\tfrac{2CD-BE}{B^2-4AC},\tfrac{2AE-BD}{B^2-4AC} ],\quad \theta^r=\arctan \left( \tfrac{C-A-\sqrt{(A-C)^2+B^2}}{B} \right),\\
  	l^{a},l^{i}=\tfrac{-\sqrt{2\left(AE^2+CD^2-BDE+(B^2-4AC)F\right) \left((A+C)\pm \sqrt{(A-C)^2+B^2}\right) }}{B^2-4AC}. \label{eq:ge3}
  \end{align}
  By substituting \eqref{eq:ellipse1} into \eqref{eq:ge1}-\eqref{eq:ge3}, we obtain \eqref{eq:center1} and \eqref{eq:axes1}.
\end{proof}

 \begin{figure}[ht]
 \centering
 	\subfigure[3D intersection of plane $h=h_f$ and cone \eqref{eq:d1ineq}]{\includegraphics[width=0.45\textwidth ]{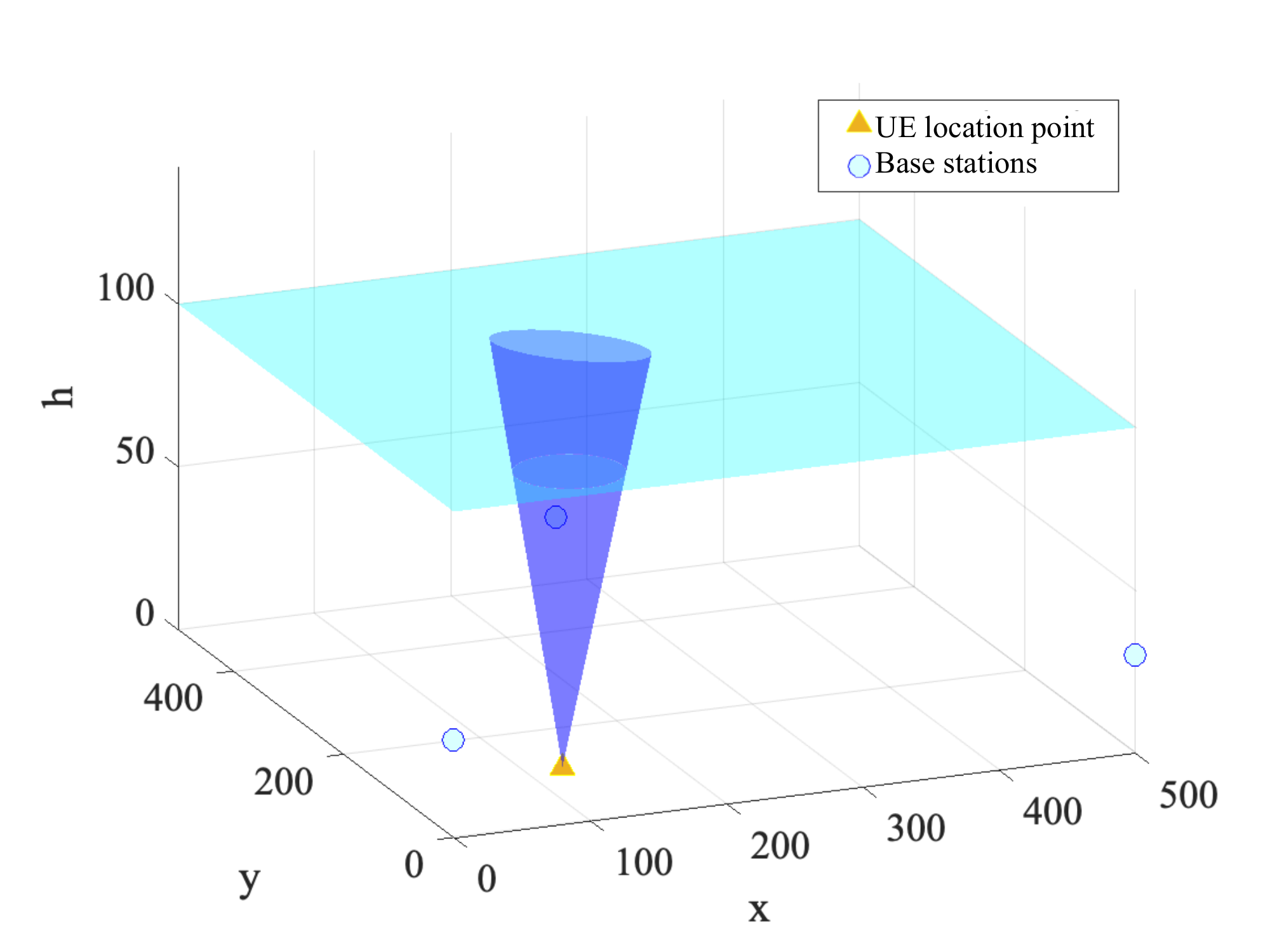}}
\subfigure[2D projection of feasible region]{\includegraphics[width=0.45\textwidth ]{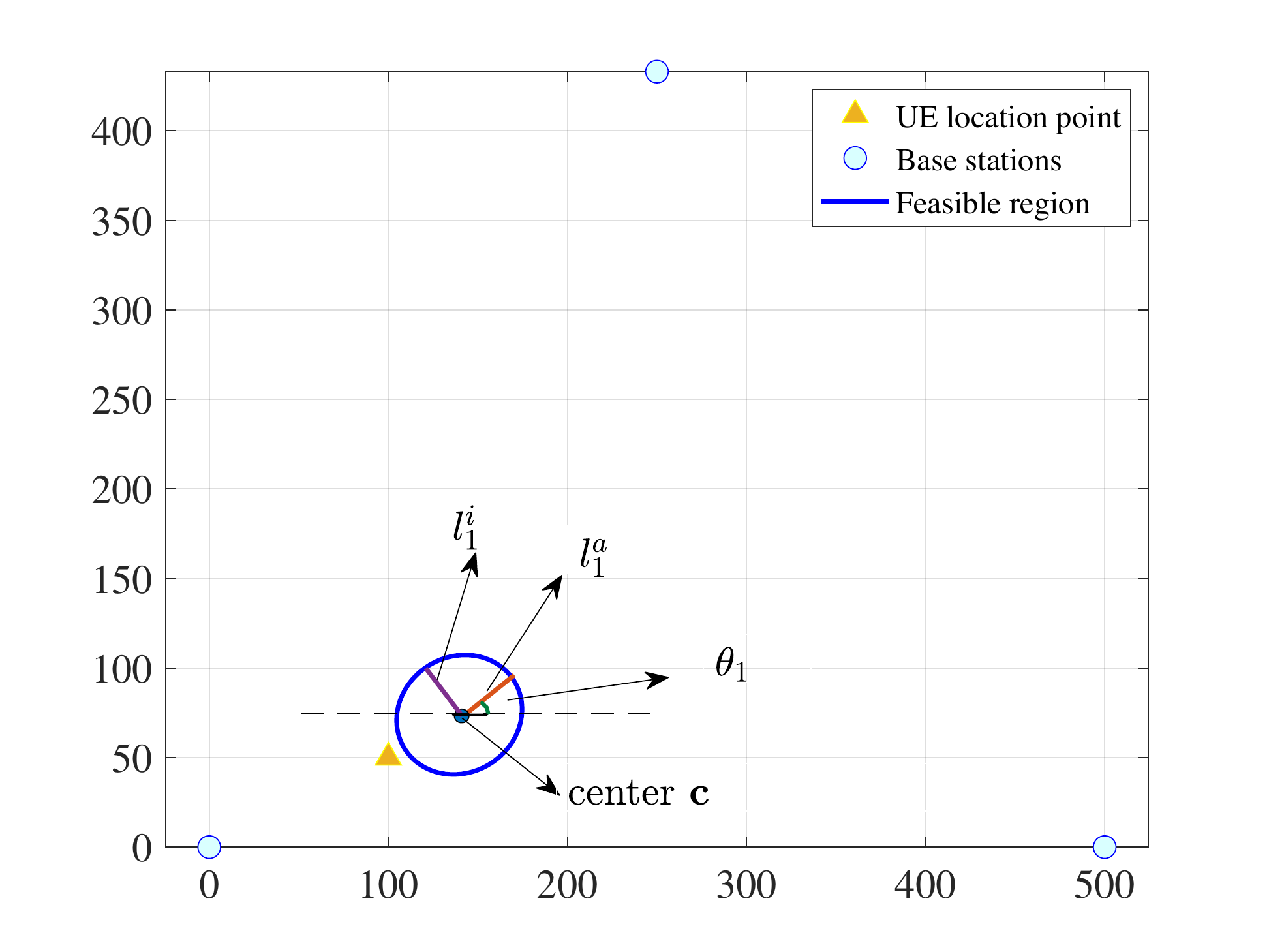}}
 	\caption{Intersections of $h_f=100$ and the extended cone when $\bm m=(100,50,10)$}
 	\label{fig:hcone}
 \end{figure} 
 
 Fig. \ref{fig:hcone}(a) shows the 3D intersection of the extended cone and the plane at $h_f=100$ m when $\bm m=(100,50,10)$. Fig. \ref{fig:hcone}(b) shows the 2D projection of the feasible UAV hovering region. We also denote the  center, the rotation angle, the semi-major axis and the semi-minor axis of the ellipse in Fig. \ref{fig:hcone}(b). In Fig. \ref{fig:hcone}(a), the UAV must be placed within the cone to satisfy the localization performance constraint. In Fig. \ref{fig:hcone}(b), given a fixed altitude, the UAV must be placed within the eclipse to meet the localization requirement. Although we derive the feasible hovering region for any given $h_f$ in Proposition \ref{prop2}, the distance between UAV and point $\bm m$ cannot be too large for two reasons. Firstly, when $h_f$ is too large,  the ellipse and its center defined in \eqref{eq:ellipse1} and \eqref{eq:center1} could deviate from the UAV restricted area $\mathcal{A}$. In this case, we cannot find a feasible UAV location in $\mathcal{A}$. Secondly, the strength of localization signal from UAV to UE becomes very weak when the distance is too large, and thus the condition for an accurate opt-$D_1$ approximation no longer holds. 
 
\subsubsection{Case 2}
  When $ \det(\mathbf{H})< 0 $, constraint \eqref{eq:locreq} is equivalent to 
\begin{align}\label{eq:case2}
 	 \alpha_1(x_u^a-x_0) +\alpha_2(y_u^a-y_0)+\alpha_3(h_u^a-h_0)   < \tilde \epsilon_2 ||\bm b_u-\bm m ||,
 \end{align}
 where $\tilde \epsilon_2=-\sqrt{D_1 \epsilon } +\alpha_1 q_{11}+\alpha_2 q_{12}+\alpha_1 q_{13}$. Similar to the first case, constraint \eqref{eq:case2} is a standard second-order cone constraint. 

 \begin{proposition}\label{prop3}
  When $\det(\mathbf{H})< 0 $, the localization requirement for UE at $\bm m$ is feasible if and only if $\epsilon$  satisfies
 	\begin{align}
 			0 \leq \epsilon \leq \tfrac{ \left(c_1+c_2\right )^2}{D_1},
 	\end{align}
 	where $c_1$ and $c_2$ are the same as in \eqref{eq:c1c2}.
 \end{proposition}
 \begin{proof}[Proof:\nopunct]
 The proof is similar to that of Proposition \ref{prop1}, and thus we omit it here. 
 \end{proof}
 
 \begin{proposition}\label{prop4}
  For a fixed altitude $h_f$, the the localization constraint \eqref{eq:case2} is satisfied if and only if the UAV $\bm b_u = [x^a_u, y^a_u,h_f]$ lies within  an ellipse defined by the following polynomial equation
 \begin{align}\label{eq:ellipse2}
 	\mathcal{E}^2:\quad &(\tfrac{\alpha_1^2 }{\tilde \epsilon_2^2}-1 )(x_u^a-x_0) ^2+\tfrac{2\alpha_1 \alpha_2 }{\tilde \epsilon_2^2} (x_u^a-x_0)(y_u^a-y_0)+(\tfrac{\alpha_2^2 }{\tilde \epsilon_2^2}-1 )(y_u^a-y_0)^2 \notag \\
 	&+\tfrac{2\alpha_1 \alpha_3 h_r }{\tilde \epsilon_2^2}(x_u^a-x_0) +\tfrac{2\alpha_2 \alpha_3 h_r }{\tilde \epsilon_2^2}(y_u^a-y_0)+ (\tfrac{\alpha_3^2}{\tilde \epsilon_2^2}-1 )h_r^2=0.
 \end{align}
The center $\bm c_2^e$ of the ellipse, the rotation angle $\theta_2$,  the semi-major axis $l_2^{a}$ and the semi-minor axis $l_2^{i}$ are respectively given by 
 \begin{align} \label{eq:center2}
 	\bm c_2^e=[\tfrac{\alpha_1 \alpha_3 h_r}{\tilde \epsilon_2^2-\alpha_1^2-\alpha_2^2}+x_0,\tfrac{\alpha_2 \alpha_3 h_r}{\tilde \epsilon_2^2-\alpha_1^2-\alpha_2^2}+y_0, h_f ], \quad
 	\theta_2&=\arctan(-\tfrac{\alpha_1}{\alpha_2} ),\\
  	l_2^{a}= \sqrt{\tfrac{h_r^2(c_1^2-\tilde \epsilon_2^2 )}{\tilde \epsilon_2^2-\alpha_1^2-\alpha_2^2 }},\quad
  	l_2^{i}=\tfrac{h_r\tilde \epsilon_2 \sqrt{c_1^2-\tilde \epsilon_2^2 }}{\tilde \epsilon_2^2-\alpha_1^2-\alpha_2^2}. \label{eq:axes2}
  \end{align}
  \end{proposition}
 \begin{proof}
 By substituting \eqref{eq:ellipse2} into \eqref{eq:ge1}-\eqref{eq:ge3}, we obtain \eqref{eq:center2} and \eqref{eq:axes2}.
 \end{proof}
 
In general, the overall feasible hovering region of a UAV can be specified by the union of two eclipses in Proposition \ref{prop2} and \ref{prop4}, or just one of them. To simplify the algorithm design to solve $\mathcal{P}1$ in the next section, we prove in the following corollary that we can set a proper value $\epsilon$ for each UE  such that the feasible hovering region of the UAV is specified by only one of the two eclipses in \eqref{eq:ellipse1} and \eqref{eq:ellipse2}.
 
\begin{corollary} By setting  $\epsilon$ such that $\left|c_3-|c_2|\right|\leq \sqrt{D_1\epsilon} \leq c_3+|c_2|$ with $c_3=\sqrt{\alpha_1^2+\alpha_2^2}$, the feasible hovering region of the UAV is $\mathcal{E}^1 $ in \eqref{eq:ellipse1} if $c_2\geq 0$. Otherwise, it is $\mathcal{E}^2$ in \eqref{eq:ellipse2} if $c_2<0$.
  \end{corollary}
 
 \begin{proof}[Proof:\nopunct]
 Please see the detailed proof in Appendix \ref{apx:col1}.
 \end{proof}

 \section{Proposed Solution} \label{sec:solution}
 
 By replacing the opt-$D$ value with opt-$D_1$, we transfer $\mathcal P1$ into the following optimization problem
\begin{subequations}\label{eq:p2}
\begin{align}
	\mathcal{P}2: \quad& \min_{\{\bm b_u \}_{u\in \mathcal U} } \quad |\mathcal U |,\\
	\text{s.t.}\quad & \max_{\{\bm b_u \}_{u\in \mathcal U}} ~  \text{opt-}D_1(\bm m_k,\bm b_u ) \geq \epsilon, \forall k \in \mathcal K,\\
	& \max_{n\in \mathcal {U}\cup\mathcal{N} }~ R_{kn} \geq R_k^{th}, \forall k \in \mathcal K,\\
	& \bm b_u \in \mathcal A, \forall u\in \mathcal U. 
\end{align}
\end{subequations}

In this section, we propose a low-complexity algorithm to solve $\mathcal P2$. In addition to the feasible localization regions specified by Corollary 1, we also need to determine the feasible UAV region to meet communication requirements. According to the communication rate of G2A channel in \eqref{eq:g2arate}, the communication requirement for UEs is equivalent to
\begin{align}\label{eq:commre1}
 ||\bm b_u-\bm m_k ||^\alpha	\leq \tfrac{P_k \hat{\mathbb{P}}( LoS,\theta_{ku})}{W_{ku}N_0\gamma_0 (2^{R_k^{th}/W_{ku}}-1 )}.
\end{align}

Because $\hat{\mathbb{P}}( LoS,\theta_{ku})$ is also a function of $||\bm b_u-\bm m_k ||$, it is difficult to depict \eqref{eq:commre1} in graph. To efficiently characterize the feasible region, we restrict the communication requirement by finding a $\hat \theta_{ku}$ such that it is irrelevant to $||\bm b_u-\bm m_k ||$ and satisfies $ \hat{\mathbb{P}}( LoS,\hat \theta_{ku})\leq \hat{\mathbb{P}}( LoS,\theta_{ku})$. 

\begin{proposition}\label{prop5}
Let 
\begin{align} \label{eq:hattheta}
	\hat \theta_{ku} = \tfrac{180}{\pi}\arcsin \left(|h_f-h_k| \big(\tfrac{  W_{ku}N_0\gamma_0 (2^{R_k^{th}/W_{ku}}-1 )  }{P_k} \big)^\frac{1}{\alpha} \right), 
\end{align}
it holds that $\hat{\mathbb{P}}( LoS,\hat \theta_{ku})\leq \hat{\mathbb{P}}( LoS,\theta_{ku}), \forall \theta_{ku} \in [0^\circ,90^\circ]$.
\end{proposition}
\begin{proof}[Proof:\nopunct] Please see the detailed proof in Appendix \ref{apx:prop5}.
\end{proof}

Then, the restricted communication requirement of UE $k$ in G2A channels becomes
\begin{align} \label{eq:circle}
 ||\bm b_u-\bm m_k ||\leq \left(\tfrac{P_k\hat{\mathbb{P}}( LoS, \hat \theta_{ku})}{W_{ku}N_0\gamma_0 (2^{R_k^{th}/W_{ku}}-1 )}  \right)^{\frac{1}{\alpha}}.
\end{align}
Notice that \eqref{eq:commre1} is automatically satisfied if \eqref{eq:circle} holds. Recall that the communication rate of G2G channel is $R_{kn}$ in \eqref{eq:g2grate}. Then, BS $n$ can meet the communication requirement of UE $k$ if its location satisfies 
\begin{align} \label{eq:bscircle}
	||\bm b_n- \bm m_k ||\leq \left ( \tfrac{F^{-1}(\varepsilon  )P_k}{W_{kn}N_0\gamma_0 (2^{R_k^{th}/W_{ku}}-1 )} \right )^{\frac{1}{\beta}}.
\end{align}

Given a fixed altitude $h_f$, the feasible region of $\text{opt-}D_1(\bm m_k,\bm b_u ) \geq \epsilon$ is an ellipse denoted by $\mathcal E_k$ and the feasible region of the communication requirement is a circle denoted by $\mathcal C_k$. We define a  region set $\mathcal{R}$ as $\mathcal{R}=\bigcup_{k=1}^K  (\mathcal C_k\cup \mathcal E_k) $, which includes all feasible regions of the communication and localization requirements. Firstly, we consider the communication requirement of each UE. If the location of UE $k$ satisfies \eqref{eq:bscircle}, BS $n$ can already serve UE $k$'s communication demand. We denote the UEs served by BSs as $\mathcal K_{BS}$ and their corresponding circles as $\mathcal C_{BS}= \bigcup_{k \in \mathcal K_{BS}}  \mathcal C_k $. Then, when considering the placement of UAVs, we can exclude  $\mathcal C_{BS}$ from $\mathcal R$. Therefore, the new set becomes  $\mathcal R_{new}= \mathcal R \backslash \mathcal C_{BS}$.  Our goal is to select a minimum-size subset $ \mathcal M= \{\bm b_u \}_{u\in \mathcal U}$ of points such that every region in $ \mathcal {R}_{new} $ has nonempty intersection with $\mathcal M$,. We recognize this as a classical minimum hitting set (MHS) problem, and summarize the solution algorithm to UAV deployment in Algorithm 1.

\begin{algorithm}[ht] \small
 \caption{Solution algorithm to UAV deployment}
    \begin{algorithmic}[1]
    \STATE \textbf{Input:} The UE locations $\{\bm m _k\}_{k=1}^K$; the fixed UAV altitude $h_f$;
      \STATE \textbf{Output:} The required number of UAVs $|\mathcal U |$; the UAV locations   $\mathcal M= \{\bm b_u \}_{u\in \mathcal U}$
         \STATE \textbf{Initialization:} transmission power, bandwidth, $R_k^{th}$, $\epsilon_k$, $\mathcal K_{BS}=\varnothing $,$\mathcal R=\varnothing $;
         \STATE Use Proposition \ref{prop1} and \ref{prop3} to check feasibility; otherwise, use Proposition \ref{prop5} to adjust the values of $\epsilon_k$;
       \FOR{$k=1$ to $K$}
         \STATE Obtain $\mathcal C_k$ according to \eqref{eq:circle} and determine $\mathcal E_k$ using \eqref{eq:ellipse1} or \eqref{eq:ellipse2};
         \STATE $\mathcal R=\mathcal R\cup \mathcal C_k \cup \mathcal E_k$;
         	\IF{$\bm m_k$ satisfies \eqref{eq:bscircle} }
						\STATE $\mathcal K_{BS}=\mathcal K_{BS} \cup \{ k\} $;
					 $\mathcal C_{BS}=\mathcal C_{BS} \cup \{ \mathcal C_k\} $;
				\ENDIF
         \ENDFOR	
         \STATE Obtain $\mathcal R_{new}= \mathcal R \backslash \mathcal C_{BS}$;
		\STATE Obtain $|\mathcal U |,\mathcal M$ by solving the MHS problem for $\mathcal R_{new}$;
			\STATE \textbf{Return:} $ |\mathcal U |,\mathcal M$    \end{algorithmic}
\end{algorithm}

 The MHS problem can be equivalently transferred into an integer liner programming (ILP) problem. The restricted area $\mathcal A$ is first discretized into a set of grid points $\mathcal{G}= \{\bm g_1, \bm g_2,...,\bm g_L \}$, which are candidates for the deployment of UAVs. Given the candidate UAV points and the $\mathcal R_{new}$, we use a binary variable $v_{ij}=1$ to denote that a grid point $\bm g_i$ is in set element $S_j \in \mathcal R_{new}$, where $S_j$ could be a circle or an ellipse region. Otherwise,  we set $v_{ij}$ to be zero. Then, the ILP problem is formulated as
	\begin{subequations}\label{eq:ilp}
	\begin{align}
	\mathcal{P}3: \quad 	\min_{v_{ij}} \quad&  \sum_{i=1}^L\sum_{j=1}^{|\mathcal R_{new}|} v_{ij} \\
		\text{s.t.} \quad &  \sum_{i=1}^L v_{ij} \geq 1, \forall S_j\in  \mathcal R_{new}.
	\end{align} 
	\end{subequations}
	
	We explain the ILP problem using a simple case with two UEs as shown in Fig. \ref{fig:example}. The red circles are the feasible communication regions and the blue ellipses are the feasible localization regions. Therefore, the region set is $\mathcal R_{new}=\{ \mathcal{C}_1,\mathcal{C}_2,\mathcal{E}_1,\mathcal{E}_2 \}$.  The small blue triangles are the candidate points $\mathcal{G}= \{\bm g_1, \bm g_2,...,\bm g_6 \}$ to deploy the UAVs.\footnote{The candidate points are randomly selected in space for illustration. In practice, the candidate points could be predetermined or calculated based on environment conditions and parameters.} The ILP problem becomes
		\begin{subequations}
	\begin{align}
   	\min_{v_{ij}} \quad&  \sum_{i=1}^6\sum_{j=1}^4 v_{ij} \\
		\text{s.t.} \quad &  \sum_{i=1}^6 v_{ij} \geq 1, \forall S_j\in  \{ \mathcal{C}_1,\mathcal{C}_2,\mathcal{E}_1,\mathcal{E}_2 \},
	\end{align} 
	\end{subequations}
and	the optimal solution is to deploy two UAVs at $\bm g_1$ and $\bm g_4$. By doing so, all the four circles/eclipses are hit by the two points.

\begin{figure}[ht]
\centering
	{\includegraphics[width=0.32\textwidth ]{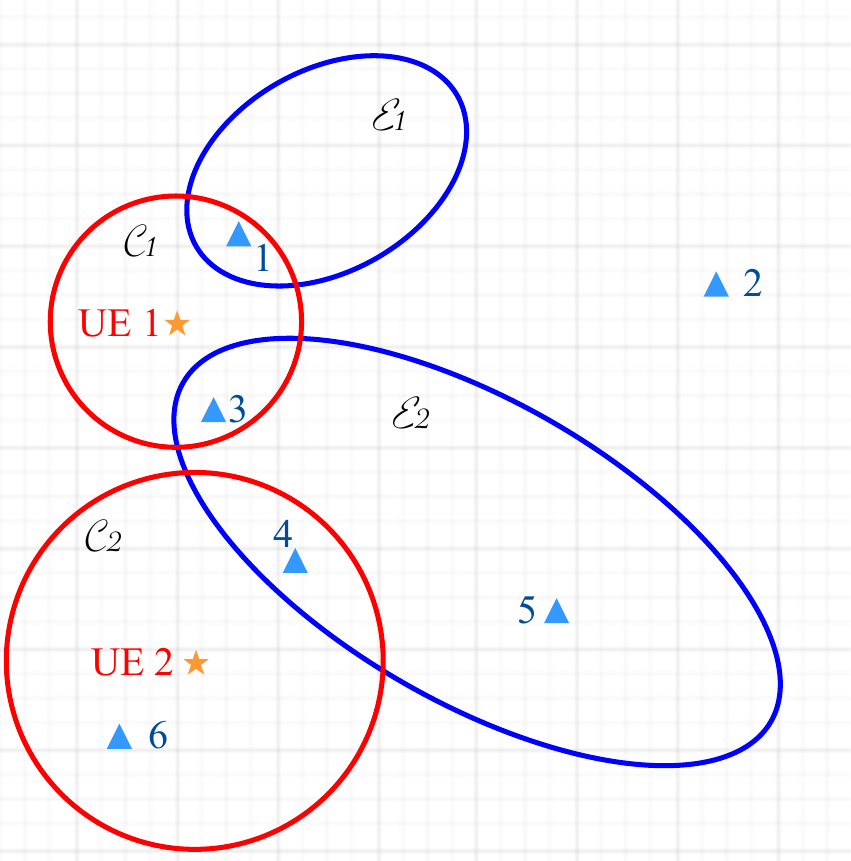}}
	\caption{Explanation with two UEs}
	\label{fig:example}
\end{figure}

	The ILP problem is NP-hard and requires high computational complexity to obtain the optimal solution. Thus, we propose a Depth-first (DF) algorithm to solve the minimum hitting set problem.  Firstly, we define the depth of point $\bm g_l$ as $\Delta_l$, which is equal to the number of regions containing point $\bm g_l$, i.e., $|\mathcal R_l |$. Then, we select the grid point with the largest depth as the first UAV location. We remove the regions containing $\bm g_l$ and update the depth of grid points. The procedure is repeated until all regions are covered by at least one UAV.  We use Fig. \ref{fig:example} to illustrate the algorithm. To start with, we observe that $\bm g_1$, $\bm g_3$ and $\bm g_4$ have the largest depth 2. Suppose that we firstly select $\bm g_1$ in $\mathcal M$. Since $\bm g_1$ is in $\mathcal C_1$ and $\mathcal E_1$, we remove these two regions from $\mathcal R_{new}$ and update the depth of grid points. Then, we see that $\bm g_4$ has the largest depth as it is inside the regions $\mathcal C_2$ and $\mathcal E_2$. Therefore, we find out that the solution is to deploy UAVs at $\bm g_1$ and $\bm g_4$. Nonetheless, the proposed solution may lead to suboptimal solution, e.g., selecting $\bm g_3$ in the first step would lead to deploy 3 UAVs at $\bm g_1$, $\bm g_3$ and $\bm g_4$. Still, we can improve the performance by executing the algorithm multiple times using randomized selection method and selecting the one produces the best solution.   
	
	The computation complexity of Algorithm 2 is $\mathcal{O}(LK+K^2)$, where $L$ is the number of grid points and $K$ is the number of UEs. The complexity can be further reduced by removing points not in any regions, i.e., $\Delta_l=0$, and the number of candidate points would be significantly deducted. The DF algorithm is summarized in Algorithm 2.

\begin{algorithm}[ht] \small
 \caption{Depth-first algorithm}
    \begin{algorithmic}[1] 
    \STATE \textbf{Input:} $\mathcal R_{new}$,  grid points $\mathcal{G}= \{\bm g_1, \bm g_2,...,\bm g_L  \}$ ;
    \STATE \textbf{Output:} $|\mathcal U |,\mathcal M $;
    \STATE \textbf{Initialization:} $\mathcal M= \varnothing$; 
    	\REPEAT 
         \FOR{$l=1$ to $L$}
         \STATE Find all regions in $\mathcal R_{new}$ containing point $\bm g_l$ and store them into $\mathcal R_l $;
         \STATE The depth of point $ \bm g_l$ is $\Delta_l=|\mathcal R_l | $;
         \ENDFOR	
         
         \STATE Find one $\bm g_l$ with the largest depth $ \Delta_l$;
         \STATE $\mathcal M= \mathcal M \cup \{\bm g_l \} $;
         \STATE Remove the regions containing $\bm g_l$ from $\mathcal R_{new}$, $\mathcal R_{new}=\mathcal R_{new}\backslash\{S_i|S_i\ni \bm g_l,\forall S_i \in  \mathcal R_{new} \} $;
       	\UNTIL $\mathcal R_{new}= \varnothing $
  		\STATE $|\mathcal U| =|\mathcal M|$;
			\STATE \textbf{Return:} $|\mathcal U |,\mathcal M$;
			   \end{algorithmic}
\end{algorithm}

\section{Simulation Results} \label{sec:simulation}

\begin{table}[ht] 
\scriptsize
\centering
\caption{Simulation Parameters}
 \begin{tabular}{||c|c||c|c||} 
 \hline\hline
 \textbf{System Parameters} & \textbf{Value} &\textbf{System Parameters} & \textbf{Value}\\ [0.5ex] 
 \hline\hline
 Main frequency, $f_c$ & 2.1 GHz  & 
Free space path loss at 1 m, $\gamma_0$ & 38.89 dB   \\
Noise power spectral density, $N_0$ & -174 dbm/Hz  &
Reference signal bandwidth, $W$ & 180 kHz  \\
G2A path loss exponent, $\alpha$ & 2 &
G2G path loss exponent, $\beta$ & 2.2 \\
Transmission power of the BSs, $P_n$ & 1 W&
Transmission power of the UAV, $P_u$ &  1 W\\
Transmission power of the UEs, $P_k$ & 0.01 W&
UE uplink transmission bandwidth, $W_{ku},W_{kn}$ &  180 kHz\\
Maximum tolerable outage probability $\varepsilon$ & 0.1&
Inverse of the CDF, $F^{-1}(\varepsilon)$ & 0.11  \\
Environmental parameter, $a$ & 15 &  
	Environmental parameter, $b$ & 0.5 \\
	Symbol duration of the OFDM signal, $T_s$ & $\tfrac{1}{1.5\times 10^4}$ s &  
	Number of subframes containing NPRS & 160  \\
	Location of BS 1, $\bm b_1$ & $[100,100,30]$ &  
	Location of BS 2, $\bm b_2$ & $[250,433,30]$ \\
	Location of BS 3, $\bm b_3$ & $[500,250,30]$ &  
	 &  \\
 [0.5ex] 
 \hline
 \end{tabular}
\end{table}

In this section, we conduct simulation experiments to verify our analysis and evaluate the performance of our proposed algorithm. The simulation parameters are summarized in Table I unless otherwise stated. We compare our proposed algorithm with the following benchmark methods:
 \subsubsection{\textbf{Integer Linear Programming (ILP) Method}} 
	The ILP problem $\mathcal P3$ is NP-hard and requires high computational complexity to obtain the optimal solution. We use the solution of ILP method as baseline to evaluate the performance of our proposed method.   
	
  \subsubsection{\textbf{Communication First Method}} In this method, we firstly deploy the UAVs the satisfy the communication requirements of all UEs. The deployed UAVs can satisfy the localization requirements of some UEs. Afterwards, we deploy additional UAVs to meet the localization demands of the remaining UEs.
  \subsubsection{\textbf{Spiral Searching Method}} In \cite{lyu2016placement}, the authors proposed a polynomial-time algorithm to sequentially deploy the UAVs along a spiral path toward the center to provide wireless communication coverage to ground users. Similarly, we deploy the UAVs to satisfy the communication and localization requirements of all UEs along a simplified spiral path shown in Fig. \ref{fig:heuristic}(a).  
  \subsubsection{\textbf{Strip Searching Method}} In \cite{agarwal2012near}, the authors proposed a near-linear-time approximation algorithm for the hitting set problem under specific geometric settings. The main idea is to use vertical decomposition to obtain overlapping regions. We adopt the strip searching as a benchmark method, which deploys the UAVs in the sequence shown in Fig. \ref{fig:heuristic}(b) from left to right.
	
	\begin{figure}[ht]
	\centering
	\subfigure[Spiral searching method]{\includegraphics[width=0.45\textwidth ]{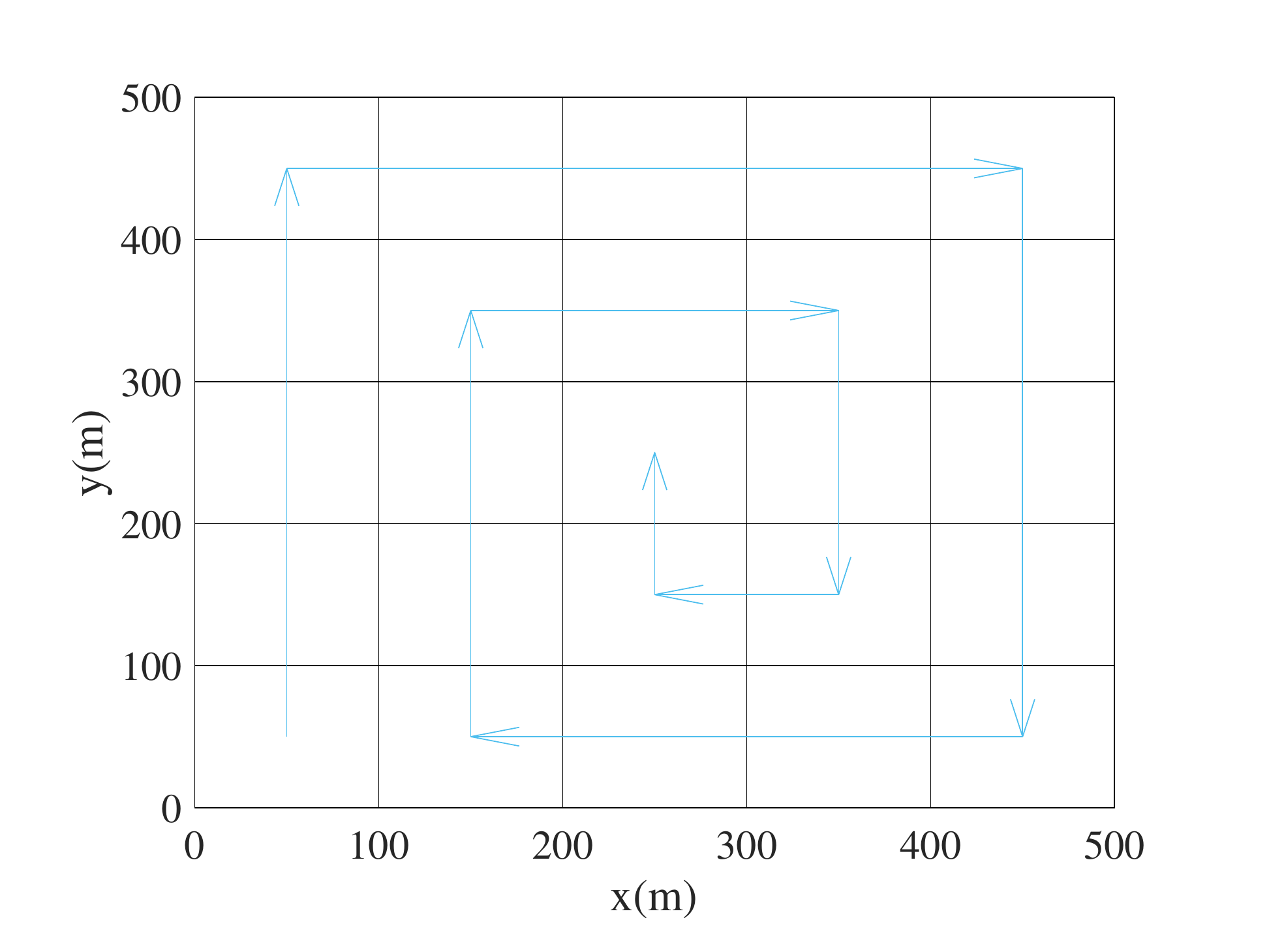}}
	\subfigure[Strip searching method]{\includegraphics[width=0.45\textwidth ]{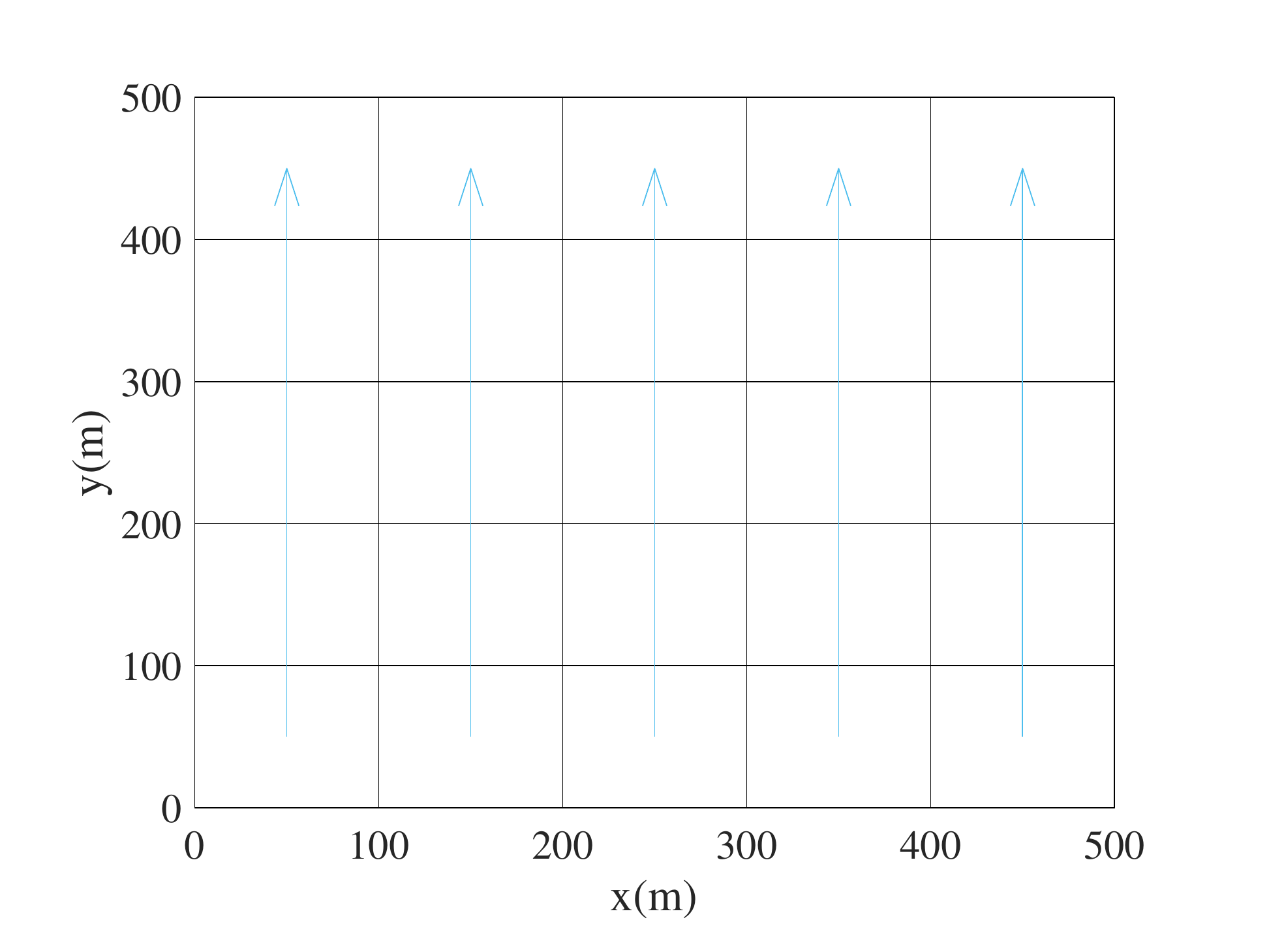}}
	
	\caption{Heuristic searching methods}
	\label{fig:heuristic}
\end{figure}

\begin{figure}
	\centering
	\subfigure[DF algorithm]{\includegraphics[width=0.45\textwidth ]{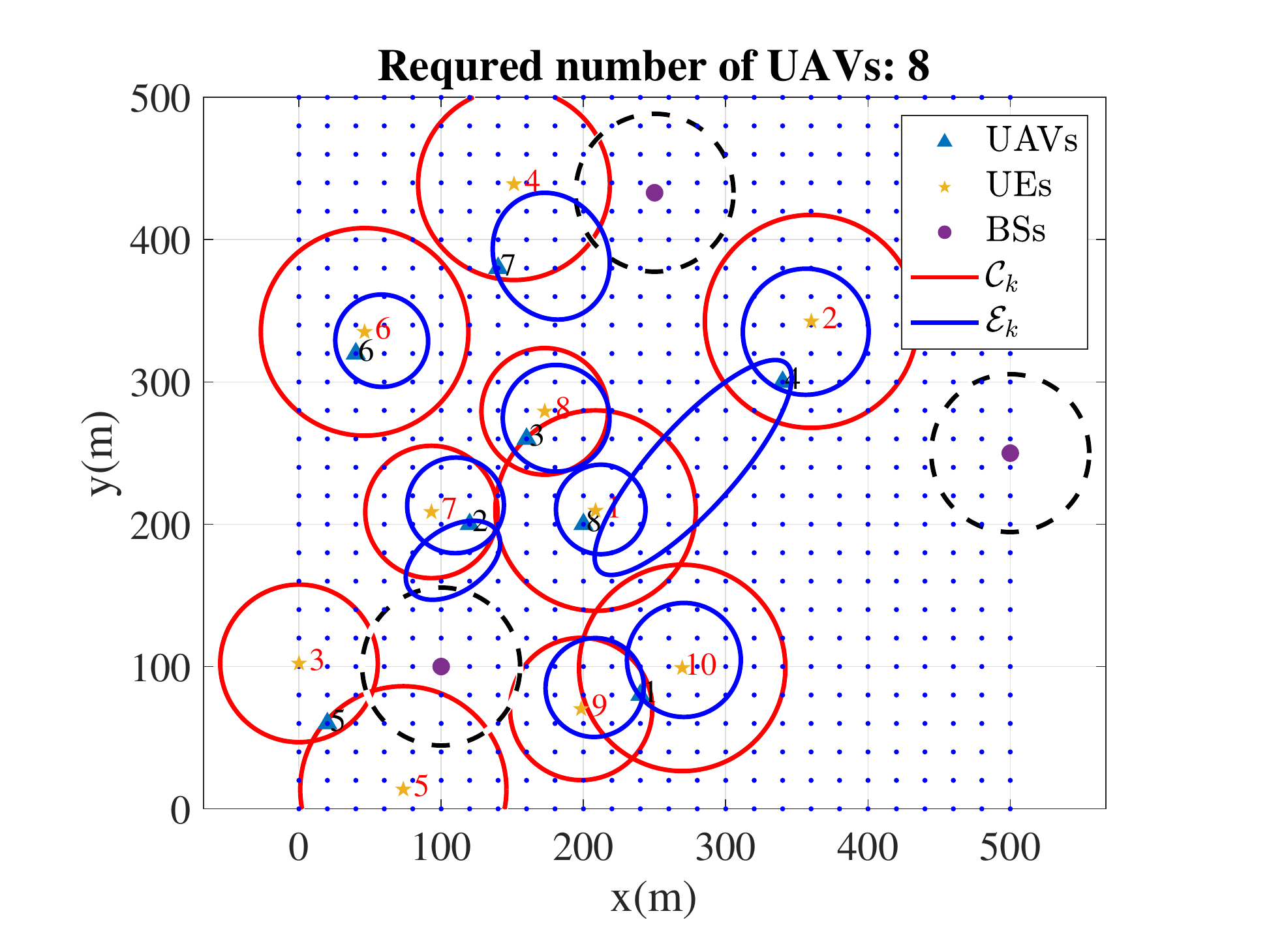}}
	\subfigure[ILP method]{\includegraphics[width=0.45\textwidth ]{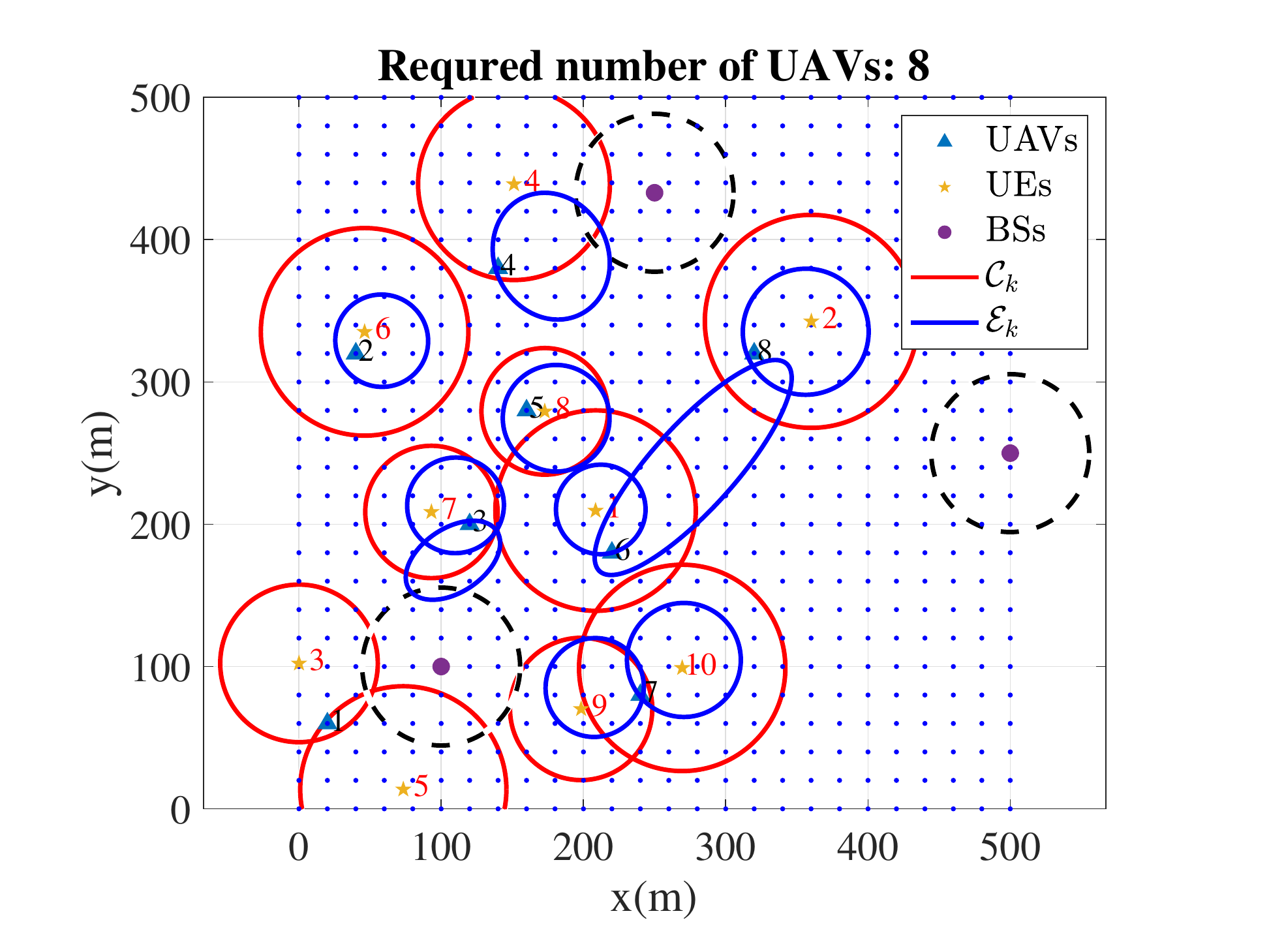}}
\centering
	\subfigure[Communication first]{\includegraphics[width=0.45\textwidth ]{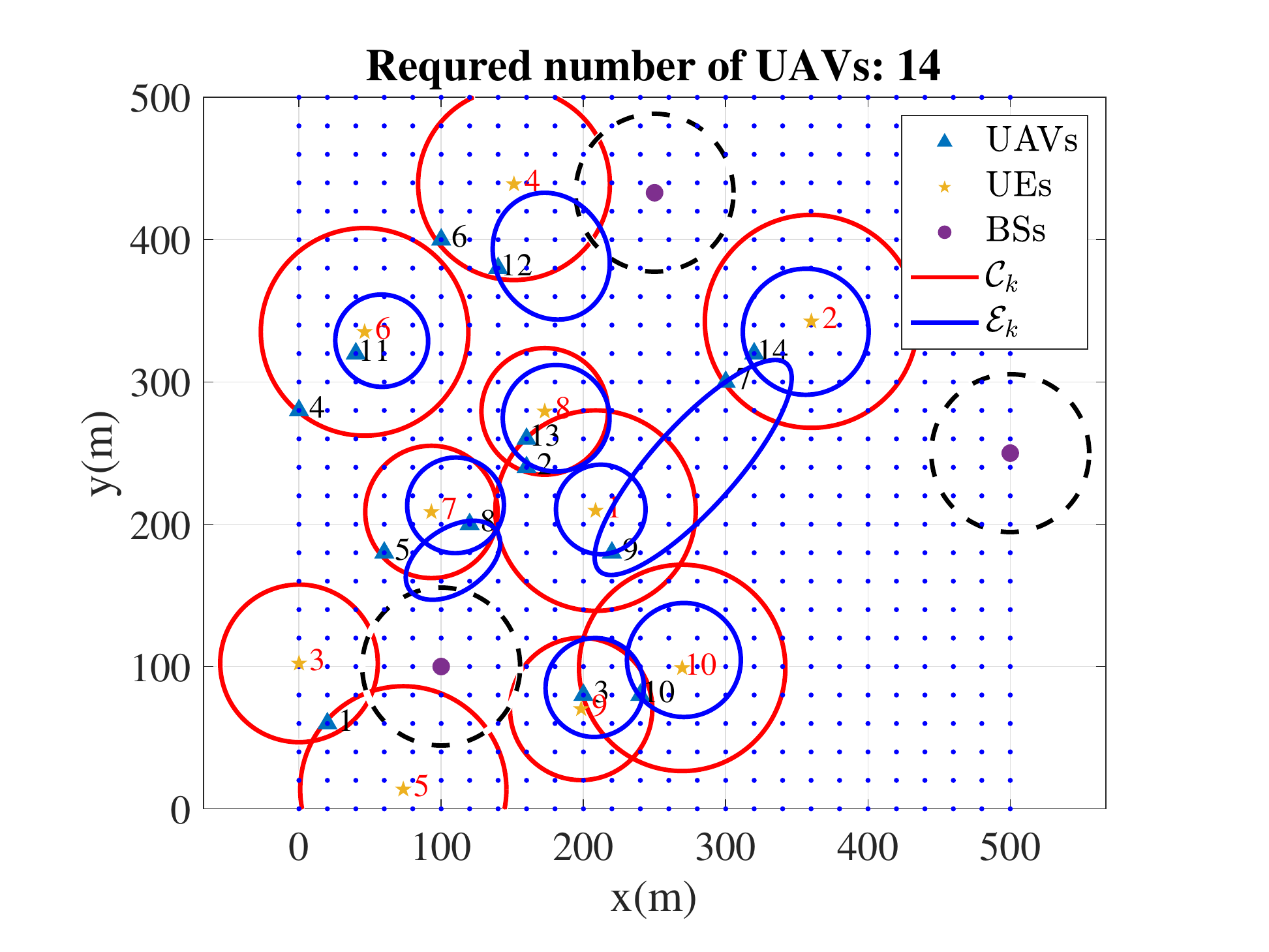}}
	\subfigure[Spiral Searching]{\includegraphics[width=0.45\textwidth ]{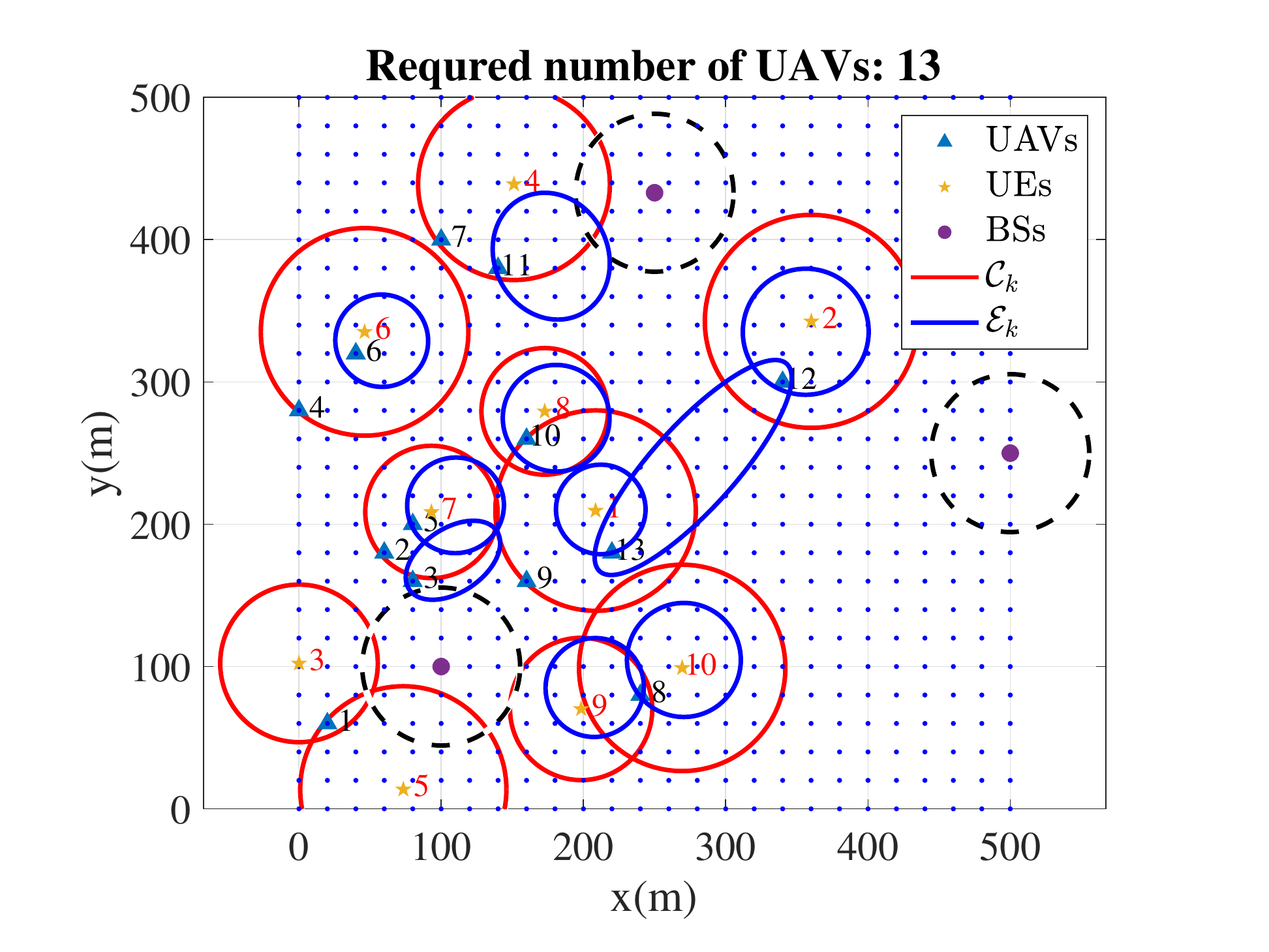}}
	\subfigure[Strip Searching]{\includegraphics[width=0.45\textwidth ]{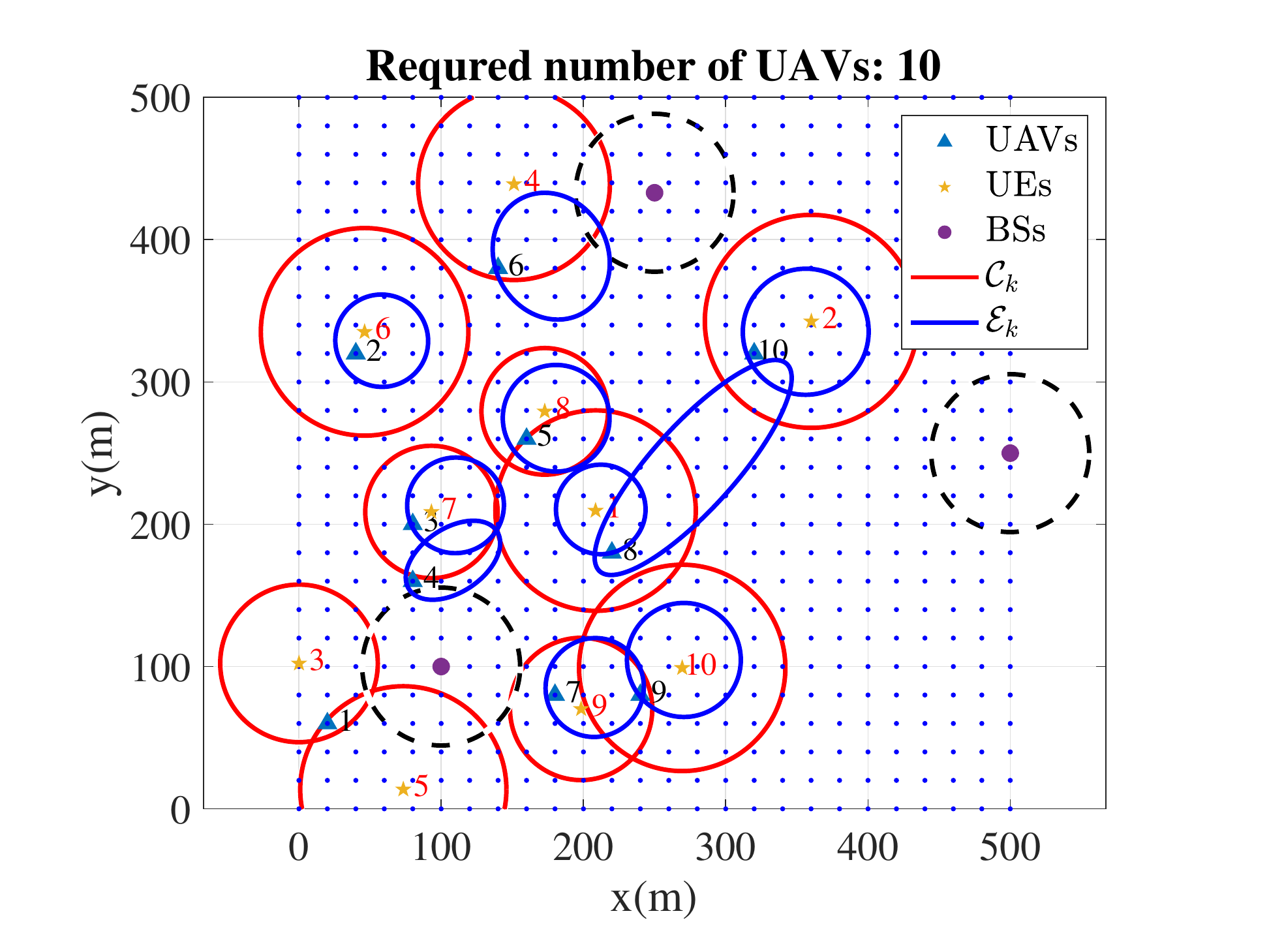}}
	\caption{Placement of UAVs under different methods}
	\label{fig:comp}
\end{figure}

Fig. \ref{fig:comp} shows the required  numbers and the placements of UAVs under different methods for 10 UEs. The localization accuracy requirement $\epsilon_k$ is randomly selected from the range specified by Corollary 1 to ensure feasibility. The communication requirement is randomly chosen from $3\leq R_k^{th} \leq 5$ Mbps. Each UE is associated with one circular region (red circle) to meet the communication performance and an elliptical region (blue ellipse) to meet the localization performance. The dashed black circles are the communication coverage of the three BSs. As long as there is a UAV hovering in its corresponding blue area, its localization accuracy requirement is satisfied. Likewise, as long as there is a UAV or ground base station in the red area, the communication requirement can be satisfied.  As shown in Fig. \ref{fig:comp}, the ILP method and the proposed DF method achieve the minimal number of UAVs, which is 8 in this scenario. The spiral searching method and the strip searching method  require 13 UAVs and  10 UAVs, respectively. The communication first method requires 14 UAVs, which is the largest among all methods. Although we use the approximation opt-$D_1$ values to determine the feasible localization regions, we find that the obtained UAV locations always satisfy the original localization requirement (\ref{eq:p1}b) in opt-$D$ values. 

 \begin{figure}[ht]
 	\centering
 	\includegraphics[width=0.45\textwidth ]{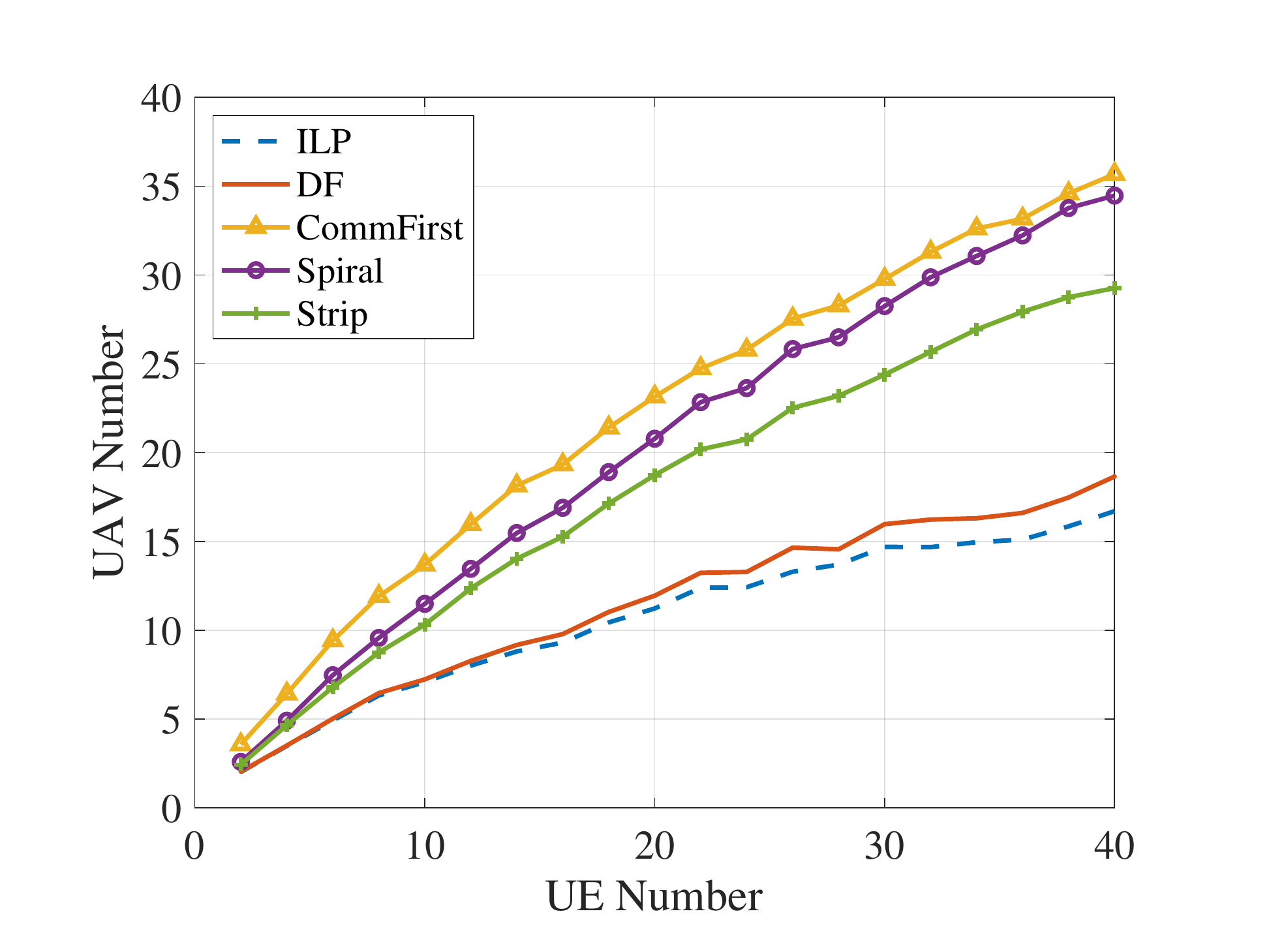}
 	\caption{UAV number v.s. UE number under different methods}
 	\label{fig:uavue}
 \end{figure}

In Fig. \ref{fig:uavue}, we show the average required number of UAVs  under all methods from 2 UEs to 40 UEs. Each sample point is the average of 100 simulation runs. The ILP method attains the optimal solutions and provides the baseline for comparison. As shown in Fig. \ref{fig:uavue}, the DF method is the closest to the ILP method. The communication first method and the spiral searching method require more than twice the optimal UAV numbers. The reason is that these methods are based on the communication coverage method. Therefore, they are not suitable for the UAV deployment for joint communication and localization scenarios. The performance of the strip searching method is slightly better than these two methods, which requires less than twice the optimal UAV numbers. When the number of UEs is 40, the average UAV number of the ILP method is $16.7$. The DF method on average needs $1.95$ more UAVs than the ILP method, while the strip searching method, the spiral searching method and the communication first method require $12.56$, $17.77$ and $18.98$ more UAVs, respectively.

\section{Conclusion} \label{sec:conclusion}

In this paper,  we investigate the optimization problem of deploying the minimal number of UAVs in feasible regions to satisfy the communication and localization requirements of the ground users. We compare the localization performance of the UAV-assisted air-ground cooperative ILAC scheme and the conventional scheme using only terrestrial networks. We observe that the scheme with UAVs can greatly improve the localization accuracy in both horizontal and vertical directions. By adopting $D$-optimality as the localization performance metric, we derive the closed-form expression and the spatial characteristics  of the feasible UAV hovering regions. We solve the deployment problem  by transforming the problem into a geometric minimum hitting set problem, and propose a low-complexity sub-optimal algorithm to solve it. We compare our proposed algorithm with benchmark methods. The number of UAVs required by the proposed algorithm is close to the optimal number, while the other benchmark methods require much more UAVs. 

\begin{appendices}

\section{Proof of Proposition \ref{prop1}} \label{apx:prop1}

With $c_1=\sqrt{\alpha_1^2 +\alpha_2^2+\alpha_3^2}$ and $c_2=\alpha_1 q_{11}+\alpha_2 q_{12}+\alpha_3 q_{13}$, we have $|c_2|\leq c_1$ according to Cauchy-Schwartz inequality:
 	 \begin{gather}
 	 	|\alpha_1 q_{11}+\alpha_2 q_{12}+\alpha_3 q_{13}|\leq \sqrt{\alpha_1^2 +\alpha_2^2+\alpha_3^2} \cdot\sqrt{q_{11}^2+q_{12}^2+q_{13}^2}\notag \\
 	 	\leq \sqrt{\alpha_1^2 +\alpha_2^2+\alpha_3^2}||\bm q_1||=\sqrt{\alpha_1^2 +\alpha_2^2+\alpha_3^2}.
 	 \end{gather}
The feasible condition for the intersection of sphere and plane is $\rho\leq 1$. Therefore, we have
\begin{gather}
    \rho=\tfrac{|\tilde \epsilon_1 |}{c_1} \leq 1 \Rightarrow |\tilde \epsilon_1 | \leq c_1 \Rightarrow|\sqrt{D_1\epsilon}+c_2| \leq c_1 
     \Rightarrow  -c_1-c_2\leq \sqrt{D_1 \epsilon}\leq c_1-c_2. \label{eq:prop1eq2}
\end{gather}
Since $|c_2|\leq c_1$, we have $-c_1-c_2\leq 0$ and \eqref{eq:prop1eq2} becomes $\sqrt{D_1 \epsilon}\leq c_1-c_2$. 

\section{Proof of Corollary 1}\label{apx:col1}
The general conic equation $Ax^2 + Bxy + Cy^2 + Dx + Ey + F = 0$ represents an ellipse if and only if
  \begin{align}\label{eq:condition}
      B^2-4AC<0.
  \end{align}
  For \eqref{eq:ellipse1}, condition \eqref{eq:condition} is equivalent to
  \begin{align}
       (\tfrac{2\alpha_1 \alpha_2 }{\tilde \epsilon_1^2})^2-4(\tfrac{\alpha_1^2 }{\tilde \epsilon_1^2}-1 )(\tfrac{\alpha_2^2 }{\tilde \epsilon_1^2}-1 ) <0\Longleftrightarrow \tfrac{\alpha_1^2+\alpha_2^2 }{\tilde \epsilon_1^2}-1<0.
    \end{align}
    Since $\tilde \epsilon_1 =\sqrt{D_1\epsilon}+c_2$ and we define $c_3=\sqrt{\alpha_1^2+\alpha_2^2 }\geq 0$, the condition is equivalent to
    \begin{align}\label{eq:cond1}
       |\sqrt{D_1\epsilon}+c_2| > c_3     \Leftrightarrow \sqrt{D_1\epsilon} >c_3-c_2, \text{~or~} \sqrt{D_1\epsilon} <-c_3-c_2,
  \end{align}
Similarly, because $\tilde \epsilon_2 = -\sqrt{D_1\epsilon}+c_2$, the condition \eqref{eq:condition}  for \eqref{eq:ellipse2} is equivalent to 
  \begin{align}\label{eq:cond2}
      &(\tfrac{2\alpha_1 \alpha_2 }{\tilde \epsilon_2^2})^2-4(\tfrac{\alpha_1^2 }{\tilde \epsilon_2^2}-1 )(\tfrac{\alpha_2^2 }{\tilde \epsilon_2^2}-1 ) <0 \Leftrightarrow |-\sqrt{D_1\epsilon}+c_2| >c_3, \notag \\
      &\Leftrightarrow \sqrt{D_1\epsilon} >c_2+c_3, \text{~or~} \sqrt{D_1\epsilon} <c_2-c_3.
  \end{align}
  When $c_2\geq 0$, $\sqrt{D_1\epsilon} <-c_3-c_2\leq 0$ in \eqref{eq:cond1} can be discarded. If $\epsilon$ satisfies $|c_3-c_2|< \sqrt{D_1\epsilon} \leq c_3+c_2$, \eqref{eq:cond1} holds while \eqref{eq:cond2} does not hold. When $c_2<0$, $\sqrt{D_1\epsilon} <c_2-c_3<0$ in \eqref{eq:cond2} can be ignored. If  $|c_3+c_2|< \sqrt{D_1\epsilon} \leq c_3-c_2$, \eqref{eq:cond2} holds while \eqref{eq:cond1} does not hold.

\section{Proof of Proposition \ref{prop5}} \label{apx:prop5}
 Since $\hat{\mathbb{P}}( LoS,\theta_{ku})\leq 1 $, we have 
 \begin{align}
 &||\bm b_u-\bm m_k ||^\alpha	\leq \tfrac{P_k \hat{\mathbb{P}}( LoS,\theta_{ku})}{W_{ku}N_0\gamma_0 (2^{R_k^{th}/W_{ku}}-1 )}\leq \tfrac{P_k}{W_{ku}N_0\gamma_0 (2^{R_k^{th}/W_{ku}}-1 )} \notag \\
 \Rightarrow &||\bm b_u-\bm m_k ||\leq \big(\tfrac{P_k}{W_{ku}N_0\gamma_0 (2^{R_k^{th}/W_{ku} }-1 )}  \big)^{\frac{1}{\alpha}}=R_0 \notag \\
 \Rightarrow & \tfrac{|h_f-h_k|}{||\bm b_u-\bm m_k ||} \geq \tfrac{|h_f-h_k|}{R_0}  \Rightarrow \sin \theta_{ku} \geq \sin \hat \theta_{ku}.
\end{align}

In addition, both $\sin \theta_{ku} $ and $\hat{\mathbb{P}}( LoS,\theta_{ku})$ are monotonically increasing when $\theta_{ku}\in [0^\circ,90^\circ] $, thus $\sin \theta_{ku} \geq \sin \hat \theta_{ku}$ implies $ \hat{\mathbb{P}}( LoS,\theta_{ku})\geq \hat{\mathbb{P}}( LoS,\hat  \theta_{ku})$. Hence, the proof is complete.
\end{appendices}

\bibliographystyle{IEEEtran}
\bibliography{IEEEabrv,references}

\begin{thebibliography}{10}
\providecommand{\url}[1]{#1}
\csname url@samestyle\endcsname
\providecommand{\newblock}{\relax}
\providecommand{\bibinfo}[2]{#2}
\providecommand{\BIBentrySTDinterwordspacing}{\spaceskip=0pt\relax}
\providecommand{\BIBentryALTinterwordstretchfactor}{4}
\providecommand{\BIBentryALTinterwordspacing}{\spaceskip=\fontdimen2\font plus
\BIBentryALTinterwordstretchfactor\fontdimen3\font minus
  \fontdimen4\font\relax}
\providecommand{\BIBforeignlanguage}[2]{{%
\expandafter\ifx\csname l@#1\endcsname\relax
\typeout{** WARNING: IEEEtran.bst: No hyphenation pattern has been}%
\typeout{** loaded for the language `#1'. Using the pattern for}%
\typeout{** the default language instead.}%
\else
\language=\csname l@#1\endcsname
\fi
#2}}
\providecommand{\BIBdecl}{\relax}
\BIBdecl

\bibitem{hu2021digital}
C.~Hu, W.~Fan, E.~Zeng, Z.~Hang, F.~Wang, L.~Qi, and M.~Z.~A. Bhuiyan,
  ``{Digital Twin-Assisted Real-Time Traffic Data Prediction Method for
  5G-Enabled Internet of Vehicles},'' \emph{IEEE Transactions on Industrial
  Informatics}, vol.~18, no.~4, pp. 2811--2819, 2021.

\bibitem{liu2021promoting}
L.~Liu, X.~Guo, and C.~Lee, ``{Promoting smart cities into the 5G era with
  multi-field Internet of Things (IoT) applications powered with advanced
  mechanical energy harvesters},'' \emph{Nano Energy}, vol.~88, p. 106304,
  2021.

\bibitem{kumhar2021emerging}
M.~Kumhar and J.~Bhatia, ``{Emerging communication technologies for 5G-Enabled
  internet of things applications},'' in \emph{Blockchain for 5G-Enabled
  IoT}.\hskip 1em plus 0.5em minus 0.4em\relax Springer, 2021, pp. 133--158.

\bibitem{hasan2021search}
M.~M. Hasan, M.~A. Rahman, A.~Sedigh, A.~U. Khasanah, A.~T. Asyhari, H.~Tao,
  and S.~A. Bakar, ``{Search and rescue operation in flooded areas: a survey on
  emerging sensor networking-enabled IoT-oriented technologies and
  applications},'' \emph{Cognitive Systems Research}, vol.~67, pp. 104--123,
  2021.

\bibitem{shahzadi2021uav}
R.~Shahzadi, M.~Ali, H.~Z. Khan, and M.~Naeem, ``{UAV assisted 5G and beyond
  wireless networks: A survey},'' \emph{Journal of Network and Computer
  Applications}, vol. 189, p. 103114, 2021.

\bibitem{zhang2018quality}
X.~Zhang, X.~Tao, F.~Zhu, X.~Shi, and F.~Wang, ``Quality assessment of gnss
  observations from an android n smartphone and positioning performance
  analysis using time-differenced filtering approach,'' \emph{Gps Solutions},
  vol.~22, no.~3, pp. 1--11, 2018.

\bibitem{wang2021toward}
Z.~Wang, R.~Liu, Q.~Liu, L.~Han, J.~S. Thompson, Y.~Lin, and W.~Mu, ``Toward
  reliable uav-enabled positioning in mountainous environments: System design
  and preliminary results,'' \emph{IEEE Transactions on Reliability}, 2021.

\bibitem{xiao2022overview}
Z.~Xiao and Y.~Zeng, ``An overview on integrated localization and communication
  towards 6g,'' \emph{Science China Information Sciences}, vol.~65, no.~3, pp.
  1--46, 2022.

\bibitem{fischer2014observed}
S.~Fischer, ``Observed time difference of arrival (otdoa) positioning in 3gpp
  lte,'' \emph{Qualcomm White Pap}, vol.~1, no.~1, pp. 1--62, 2014.

\bibitem{chen2016three}
C.-Y. Chen and W.-R. Wu, ``Three-dimensional positioning for lte systems,''
  \emph{IEEE Transactions on vehicular technology}, vol.~66, no.~4, pp.
  3220--3234, 2016.

\bibitem{al2014optimal}
A.~Al-Hourani, S.~Kandeepan, and S.~Lardner, ``Optimal lap altitude for maximum
  coverage,'' \emph{IEEE Wireless Communications Letters}, vol.~3, no.~6, pp.
  569--572, 2014.

\bibitem{chen2017optimum}
Y.~Chen, W.~Feng, and G.~Zheng, ``Optimum placement of uav as relays,''
  \emph{IEEE Communications Letters}, vol.~22, no.~2, pp. 248--251, 2017.

\bibitem{lyu2016placement}
J.~Lyu, Y.~Zeng, R.~Zhang, and T.~J. Lim, ``Placement optimization of
  uav-mounted mobile base stations,'' \emph{IEEE Communications Letters},
  vol.~21, no.~3, pp. 604--607, 2016.

\bibitem{ali2020uav}
M.~A. Ali and A.~Jamalipour, ``Uav placement and power allocation in uplink and
  downlink operations of cellular network,'' \emph{IEEE Transactions on
  Communications}, vol.~68, no.~7, pp. 4383--4393, 2020.

\bibitem{ren2020joint}
H.~Ren, C.~Pan, K.~Wang, W.~Xu, M.~Elkashlan, and A.~Nallanathan, ``Joint
  transmit power and placement optimization for urllc-enabled uav relay
  systems,'' \emph{IEEE Transactions on vehicular technology}, vol.~69, no.~7,
  pp. 8003--8007, 2020.

\bibitem{fan2018optimal}
R.~Fan, J.~Cui, S.~Jin, K.~Yang, and J.~An, ``Optimal node placement and
  resource allocation for uav relaying network,'' \emph{IEEE Communications
  Letters}, vol.~22, no.~4, pp. 808--811, 2018.

\bibitem{sabzehali20213d}
J.~Sabzehali, V.~K. Shah, H.~S. Dhillon, and J.~H. Reed, ``{3D Placement and
  Orientation of mmWave-based UAVs for Guaranteed LoS Coverage},'' \emph{IEEE
  Wireless Communications Letters}, vol.~10, no.~8, pp. 1662--1666, 2021.

\bibitem{8226757}
J.~A. del Peral-Rosado, R.~Raulefs, J.~A. López-Salcedo, and G.~Seco-Granados,
  ``{Survey of Cellular Mobile Radio Localization Methods: From 1G to 5G},''
  \emph{IEEE Communications Surveys \& Tutorials}, vol.~20, no.~2, pp.
  1124--1148, 2018.

\bibitem{liu2018spectrum}
Y.~Liu, J.~Wang, and Y.~Shen, ``{Spectrum allocation for UAV-aided relative
  localization of ground vehicles},'' in \emph{2018 IEEE International
  Conference on Communications Workshops (ICC Workshops)}.\hskip 1em plus 0.5em
  minus 0.4em\relax IEEE, 2018, pp. 1--5.

\bibitem{wu2019resource}
Z.~Wu, W.~Quan, and T.~Zhang, ``{Resource allocation in UAV-aided vehicle
  localization frameworks},'' in \emph{2019 IEEE/CIC International Conference
  on Communications Workshops in China (ICCC Workshops)}.\hskip 1em plus 0.5em
  minus 0.4em\relax IEEE, 2019, pp. 98--103.

\bibitem{sorbelli2018range}
F.~B. Sorbelli, S.~K. Das, C.~M. Pinotti, and S.~Silvestri, ``{Range based
  algorithms for precise localization of terrestrial objects using a drone},''
  \emph{Pervasive and Mobile Computing}, vol.~48, pp. 20--42, 2018.

\bibitem{yang2021deployment}
J.~Yang, T.~Liang, and T.~Zhang, ``{Deployment Optimization in UAV Aided
  Vehicle Localization},'' in \emph{2021 IEEE 93rd Vehicular Technology
  Conference (VTC2021-Spring)}.\hskip 1em plus 0.5em minus 0.4em\relax IEEE,
  2021, pp. 1--6.

\bibitem{esrafilian2020three}
O.~Esrafilian, R.~Gangula, and D.~Gesbert, ``{Three-dimensional-map-based
  trajectory design in UAV-aided wireless localization systems},'' \emph{IEEE
  Internet of Things Journal}, vol.~8, no.~12, pp. 9894--9904, 2020.

\bibitem{zhan2017energy}
C.~Zhan, Y.~Zeng, and R.~Zhang, ``{Energy-efficient data collection in UAV
  enabled wireless sensor network},'' \emph{IEEE Wireless Communications
  Letters}, vol.~7, no.~3, pp. 328--331, 2017.

\bibitem{9636988}
A.~W. Al-Asadi and N.~S. Ali, ``{Noise-Robust Least-Squares Method in TDOA
  Estimation of a Source Location},'' in \emph{2021 Palestinian International
  Conference on Information and Communication Technology (PICICT)}, 2021, pp.
  123--128.

\bibitem{del2012achievable}
J.~A. del Peral-Rosado, J.~A. L{\'o}pez-Salcedo, G.~Seco-Granados, F.~Zanier,
  and M.~Crisci, ``{Achievable localization accuracy of the positioning
  reference signal of 3GPP LTE},'' in \emph{2012 International Conference on
  Localization and GNSS}.\hskip 1em plus 0.5em minus 0.4em\relax IEEE, 2012,
  pp. 1--6.

\bibitem{wang2019energy}
Z.~Wang, R.~Liu, Q.~Liu, J.~S. Thompson, and M.~Kadoch, ``{Energy-efficient
  data collection and device positioning in UAV-assisted IoT},'' \emph{IEEE
  Internet of Things Journal}, vol.~7, no.~2, pp. 1122--1139, 2019.

\bibitem{ucinski2004optimal}
D.~Ucinski, \emph{Optimal measurement methods for distributed parameter system
  identification}.\hskip 1em plus 0.5em minus 0.4em\relax CRC press, 2004.

\bibitem{zheng2018oparray}
Y.~Zheng, M.~Sheng, J.~Liu, and J.~Li, ``Oparray: Exploiting array orientation
  for accurate indoor localization,'' \emph{IEEE Transactions on
  Communications}, vol.~67, no.~1, pp. 847--858, 2018.

\bibitem{agarwal2012near}
P.~K. Agarwal, E.~Ezra, and M.~Sharir, ``Near-linear approximation algorithms
  for geometric hitting sets,'' \emph{Algorithmica}, vol.~63, no.~1, pp. 1--25,
  2012.

\end{thebibliography}

\end{document}